\documentclass[11pt]{article}

\usepackage[margin=1in]{geometry}
\usepackage{amsfonts,amssymb,amsmath,amsthm,mathtools}
\usepackage{enumitem,lpic,tocloft,booktabs}
\usepackage[bf]{caption}

\usepackage{thmtools,thm-restate}
\declaretheorem[name=Theorem]{theorem}
\declaretheorem[name=Lemma,sibling=theorem]{lemma}
\declaretheorem[name=Proposition,sibling=theorem]{proposition}
\declaretheorem[name=Fact,sibling=theorem]{fact}
\declaretheorem[name=Claim,sibling=theorem]{claim}
\declaretheorem[name=Definition,style=definition]{definition}

\usepackage{complexity,bm,xspace}
\usepackage[usenames,dvipsnames]{xcolor}

\usepackage{hyperref}
\hypersetup{
	colorlinks=true,
	linkcolor=Sepia,
	citecolor=Sepia,
	filecolor=Sepia,
	urlcolor=Sepia
}

\usepackage[greek,english]{babel}
\addto\extrasenglish{}
\addto\extrasenglish{}

\usepackage{tikz,ifthen,etoolbox,color}
\usetikzlibrary{arrows.meta}

\usepackage[framemethod=tikz]{mdframed}
\mdfsetup{
	roundcorner=4pt,
	nobreak=true,
	linewidth=.5,
	innerleftmargin=25pt,
	innerrightmargin=25pt,
	innertopmargin=15pt,
	innerbottommargin=15pt,
	skipabove=12,skipbelow=8pt
}

\newcommand{\Var}{\mathop{\mathrm{Var}}}
\renewcommand{\E}{\mathop{\mathbb{E}}}

\renewcommand{\poly}{\mathrm{poly}}
\renewcommand{\polylog}{\mathrm{polylog}}

\DeclareMathOperator{\supp}{supp}
\DeclareMathOperator{\cbs}{cbs}
\DeclareMathOperator{\bs}{bs}
\DeclareMathOperator{\indeg}{in-deg}
\DeclareMathOperator{\outdeg}{out-deg}

\newcommand{\del}{\textsl{\upshape del}}
\newcommand{\cut}{\textsl{\upshape cut}}
\newcommand{\Enc}{\mathrm{Enc}}

\newcommand{\calA}{\mathcal{A}}
\newcommand{\calB}{\mathcal{B}}
\newcommand{\calD}{\mathcal{D}}
\newcommand{\calP}{\mathcal{P}}
\newcommand{\calH}{\mathcal{H}}
\newcommand{\calG}{\mathcal{G}}
\newcommand{\calO}{\mathcal{O}}

\newcommand{\tOmega}{\tilde{\Omega}}

\newcommand{\tO}{\tilde{O}}

\newcommand{\EoL}{\textsc{EoL}\xspace}
\newcommand{\QEoL}{\textsc{Q-EoL}\xspace}
\newcommand{\EndofLine}{\textsc{End-of-Line}\xspace}
\newcommand{\BFP}{\textsc{BFP}\xspace}

\newcommand{\eps}{\texorpdfstring{$\epsilon$}{eps}}

\begin{document}

\newgeometry{margin=1.3in,top=1.5in,bottom=1in}

\begin{center}
{\LARGE Near-Optimal Communication Lower Bounds for\\[1mm] Approximate Nash Equilibria}
\\[1cm] \large

\setlength\tabcolsep{2em}
\begin{tabular}{cccc}
Mika G\"o\"os &
Aviad Rubinstein\\[-.5mm]
\small\slshape Harvard University&
\small\slshape Harvard University
\end{tabular}

\vspace{7mm}

\large
{\today}

\vspace{10mm}

\normalsize
\bf Abstract
\end{center}

\noindent
We prove an $N^{2-o(1)}$ lower bound on the randomized communication complexity of finding an $\epsilon$-approximate Nash equilibrium (for constant $\epsilon>0$) in a two-player $N\times N$ game.

\vspace{10mm}

\setlength{\cftbeforesecskip}{0pt}
\setcounter{tocdepth}{2}
\tableofcontents

\thispagestyle{empty}
\setcounter{page}{0}

\newpage
\restoregeometry

\section{Introduction}

How many bits of communication are needed to find an $\epsilon$-approximate
Nash equilibrium (for a small constant $\epsilon>0$) in a two-player $N\times N$ game? More precisely:
\begin{itemize}[itemsep=0pt]
\item {\bf Alice} holds her payoff matrix $A\in[0,1]^{N\times N}$ of some constant precision.
\item {\bf Bob} holds his payoff matrix $B\in[0,1]^{N\times N}$ of some constant precision.
\item {\bf Output} an\emph{ $\epsilon$-approximate Nash equilibrium}: a mixed strategy $\calA$ for Alice and a mixed strategy $\calB$ for Bob such that neither player can unilaterally change their strategy and increase their expected payoff by more than $\epsilon$. (See \autoref{sec:def-nash} for a formal definition.)
\end{itemize}
It is well known that such approximate equilibria have a concise $O(\log^2 N)$-bit description: one may assume wlog that $\calA$ and $\calB$ are supported on at most $O(\log N)$ actions~\cite{LMM}. There is a trivial upper bound of $O(N^2)$ by communicating an entire payoff matrix. Previous work \cite{BR16-CC-nash} showed that finding an $\epsilon$-approximate Nash equilibrium requires $N^{\delta}$ bits of communication for a small constant $\delta>0$. In this work, we improve this to a near-optimal $N^{2-o(1)}$ lower bound. Our main theorem is slightly more general, as it also applies for games of asymmetric dimensions. 
\begin{theorem} \label{thm:main}
There exists an $\epsilon>0$ such that for any constants $a,b>0$ the randomized communication complexity of finding an $\epsilon$-approximate Nash equilibrium in an $N^{a}\times N^{b}$ game is $N^{a+b-o\left(1\right)}$.
\end{theorem}

It is interesting to note that there is an $\tO(N)$ protocol for computing an $\epsilon$-correlated equilibrium of an $N\times N$ game~\cite{BM07-regret}. Hence our result is the first that separates approximate Nash and correlated equilibrium.

Our result also implies the first near-quadratic lower bound for finding an approximate Nash equilibrium in the weaker query complexity model, where the algorithm has black-box oracle access to the payoff matrices (previous work established such lower bounds only against deterministic algorithms \cite{FS16-query}). In this query complexity model, there is an $\tO(N)$-queries algorithm for computing an $\epsilon$-coarse correlated equilibrium of an $N\times N$ game~\cite{GR16-query}. Hence our result is the first that separates approximate Nash and coarse correlated equilibrium in the query complexity model. See \autoref{table:results} for a summary of known bounds\footnote{We thank Yakov Babichenko for his help in understanding these connections and other insightful communication.}.

\subsection{Background}

Nash equilibrium is the central solution concept in game theory. It
is named after John Nash who, more than 60 years ago, proved that
every game has an equilibrium \cite{Nash}. Once players are at an
equilibrium, they do not have an incentive to deviate. However, Nash's
theorem does not explain how the players arrive at an equilibrium
in the first place. 

Over the last several decades, many {\em dynamics}, or procedures
by which players in a repeated game update their respective strategies
to adapt to other players' strategies, have been proposed since Nash's
result (e.g., \cite{Brown51-fictitious_play,Robinson51-fictitious_play,KaLe,HM03-deterministic_dynamics,FY06-converging_dynamics}).
But despite significant effort, we do not know any plausible dynamics
that converge even to an approximate Nash equilibrium. It is thus
natural to conjecture that there are no such dynamics. However, one
has to be careful about defining ``plausible'' dynamics. The first
example of dynamics we consider {\em im}plausible, are ``players
agree a priori on a Nash equilibrium''. The {\em uncoupled dynamics}
model proposed by Hart and Mas-Collel \cite{HM03-deterministic_dynamics}
rules out such trivialities by requiring that a player's strategy
depends only on her own utility function and the history of other
players' actions. Another example of implausible dynamics that converge
to a Nash equilibrium are exhaustive search dynamics that enumerate
over the entire search space. (Exhaustive search dynamics can converge to an approximate Nash equilibrium in finite time by enumerating over an $\epsilon$-net of the search space.) We thus consider a second natural desideratum, which is that dynamics
should converge (much) faster than exhaustive search. Note that the
two restrictions (uncoulpled-ness and fast convergence) are still
very minimal---it is still easy to come up with dynamics
that satisfy both and yet do not plausibly expect to predict players'
behavior. But, since we are after an impossibility result, it is fair
to say that if we can rule out any dynamics that satisfy these two
restrictions and converge to a Nash equilibrium, we have strong evidence
against any plausible dynamics.

A beautiful observation by Conitzer and Sandholm \cite{CS} and Hart
and Mansour \cite{HM} is that the {\em communication complexity}
of computing an (approximate) Nash equilibrium, in the natural setting
where each player knows her own utility function, precisely captures
(up to a logarithmic factor) the number of rounds for an arbitrary
uncoupled dynamics to converge to an (approximate) Nash equilibrium.
Thus the question of ruling out plausible dynamics is reduced to the
question of proving lower bounds on communication complexity. There are also other good reasons to study the communication complexity
of approximate Nash equilibria; see e.g.~\cite{Roughgarden14-barriers-to-equilibria}.

\begin{table}
\centering
\renewcommand{\arraystretch}{1.15}
\begin{tabular}{l@{\hspace{3mm}}rl@{\hspace{5mm}}rl}
\toprule[.5mm]
\bf
Type of equilibrium
& \multicolumn{2}{c}{\bf Query Complexity}
& \multicolumn{2}{c}{\bf Communication Complexity} \\
\midrule
$\epsilon$-Nash equilibirum &
$\Omega(N^{2-o(1)})$ & [This paper] &
\hspace{9mm}$\Omega(N^{2-o(1)})$ & [This paper] \\
$\epsilon$-correlated equilibrium &
$\Omega(N)$ & [Folklore] &
$O(N\log N)$ & \cite{BM07-regret} \\
$\epsilon$-coarse correlated equilibrium &
$\tilde{\Theta}(N)$ & \cite{GR16-query} &
$\polylog(N)$ & [Folklore] \\
\bottomrule[.5mm]
\end{tabular}
\caption{Query and communication complexities of approximate equilibria. The oracle query model is more restrictive than the communication model. Nash equilibrium is more restrictive than correlated equilibrium, which is yet more restrictive than coarse correlated equilibrium. Hence the complexity of a problem increases as we move up and left in the table. For all problems there are trivial
bounds of $\Omega(\log N)$ and $O(N^2)$.}
\label{table:results}
\end{table}

\subsection{Related work}

The problem of computing (approximate) Nash equilibrium has been studied
extensively, mostly in three models: communication complexity, query
complexity, and computational complexity.

\paragraph{Communication complexity.}
The study of the communication complexity of Nash equilibria was
initiated by Conitzer and Sandholm \cite{CS} who proved a quadratic
lower on the communication complexity of deciding whether a game has
a pure equilibrium, even for zero-one payoff (note that for pure equilibrium
this also rules out any more efficient approximation). Hart and Mansour
\cite{HM} proved exponential lower bounds for pure and exact mixed
Nash equilibrium in $n$-player games. Roughgarden and Weinstein \cite{RW}
proved communication complexity lower bounds on the related problem
of finding an approximate Brouwer fixed point (on a grid). In \cite{BR16-CC-nash},
in addition to the lower bound for two-player game which we improved,
there is also an exponential lower bound for $n$-player games. The
same paper also posed the open problem of settling the communication
complexity of approximate correlated equilibrium in two-player games;
partial progress has been made by \cite{GS17-correlated,KS17-correlated},
but to date the problem of determining the communication complexity
of $\epsilon$-approximate correlated equilibrium remains open
(even at the granularity of $\poly(N)$ vs $\polylog(N)$).
On the algorithmic side, Czumaj et al.~\cite{CDFFJS15-0.382} gave
a polylogarithmic protocol for computing a $0.382$-approximate Nash
equilibrium in two-player games, improving upon~\cite{GP}.

\paragraph{Query complexity.}
In the query complexity model, the algorithm has black-box oracle
access to the payoff matrix of each player. Notice that this model
is strictly weaker than the communication complexity model (hence
our communication lower bound applies to this model as well). For
the {\em deterministic} query complexity of $\epsilon$-approximate
Nash equilibrium in two-player games, Fearnley et al.~\cite{FGGS}
proved a linear lower bound, which was subsequently improved to (tight)
quadratic by Fearnley and Savani~\cite{FS16-query}. For randomized
algorithms, the only previous lower bound was $\Omega(1/\epsilon^2)$,
also by \cite{FS16-query}; notice that this is only interesting for
$\epsilon=o(1)$. As mentioned earlier, Goldberg and
Roth~\cite{GR16-query} can find an approximate coarse correlated
equilibrium with $\tO(N)$ queries (two-player $N\times N$
game) or $\tO(n)$ queries ($n$-player game
with two actions per player). In the latter regime of $n$-player
and two action per player, a long sequence of works \cite{HN13,Bab12-completely_uncoupled,CCT,R16}
eventually established an $2^{\Omega(n)}$ lower bound
on the query complexity of $\epsilon$-approximate Nash equilibrium.
(This last result is implied by and was the starting point for the
exponential query complexity lower bound in \cite{BR16-CC-nash}).
Finally, some of the aforementioned result are inspired by a query
complexity lower bound for approximate fixed point due to Hirsch et
al.~\cite{HPV89} and its adaptation to $\ell_2$-norm in \cite{R16};
this construction will also be the starting point of our reduction.

\paragraph{Computational complexity.}
For computational complexity, the problem of finding an exact Nash
equilibrium in a two-player game is \PPAD-complete~\cite{DGP,CDT}. Following a sequence
of improvements \cite{KPS09-nash-0.75,DMP09-nash-0.5,DMP07-nash-0.38,BBM10-0.36392,TS08-0.33},
we know that a $0.3393$-approximate Nash equilibrium can be computed
in polynomial time. But there exists some constant $\epsilon>0$,
such that assuming the ``Exponential Time Hypothesis (ETH) for \PPAD'',
computing an $\epsilon$-approximate Nash equilibrium requires
$N^{\log^{1-o(1)}N}$ time \cite{R16}, which is essentially
tight by \cite{LMM}.

\subsection{Definition of \eps-Nash} \label{sec:def-nash}

A two-player game is defined by two utility functions (or payoff matrices) $U^A,U^B\colon S_A\times S_B\to[0,1]$. A mixed strategy for Alice (resp.\ Bob) is a distribution $\calA$ ($\calB$) over the set of actions $S_A$ ($S_B$). We say that $(\calA,\calB)$ is an \emph{$\epsilon$-approximate Nash equilibrium} ($\epsilon$-ANE) if every alternative Alice-strategy $\calA'$ performs at most $\epsilon$ better than $\calA$ against Bob's mixed strategy $\calB$, and the same holds with roles reversed. Formally, the condition for Alice is
\[
\E_{\substack{a\sim\calA \\ b\sim\calB}}\left[U^A(a,b)\right]~\geq~
\max_{\calA'\text{ over }S_A}\enspace\E_{\substack{a'\sim\calA' \\ b\sim\calB}}\left[U^A(a',b)\right]-\epsilon.
\]

\section{Technical Overview}

Our proof follows the high-level approach of \cite{BR16-CC-nash}; see also the lecture notes~\cite{Roughgarden18-barbados} for exposition. The approach of \cite{BR16-CC-nash} consist of four steps. 
\begin{enumerate}[noitemsep]
\item Query complexity lower bound for the well-known \PPAD-complete \EoL problem.
\item Lifting of the above into a communication lower bound for a two-party version of \EoL.
\item Reduction from \EoL to $\epsilon$-\BFP, a problem of finding an (approximate) Brouwer fixed point.
\item Constructing a hard two-player game that combines problems from both Step 2 and Step 3.
\end{enumerate}
In this paper we improve the result from~\cite{BR16-CC-nash} by optimizing Steps 1, 2, and 4.

\subsection{Steps 1--2: Lower bound for End-of-Line} \label{sec:steps-12}

The goal of Steps 1--2 is to obtain a randomized communication lower bound for the \EndofLine (or \EoL for short) problem: Given an implicitly described graph on $[N]$ where a special vertex $1\in[N]$ is the start vertex of a path, find an end of a path or a non-special start of a path. The following definition is a ``template'' in that it does not yet specify the protocols $\Pi_v$.
\begin{mdframed} \label{def:eol}
\begin{center}
\large \EoL template
\end{center}
\begin{itemize}[leftmargin=2mm]
\item {\bf Input:}~~Alice and Bob receive inputs $\alpha$ and $\beta$ that implicitly describe \emph{successor} and \emph{predecessor} functions $S,P\colon[N]\to[N]$. Namely, for each $v\in[N]$ there is a ``low-cost'' protocol $\Pi_v(\alpha,\beta)$ to compute the pair $(S(v),P(v))$.
\item {\bf Output:}~~Define a digraph $G=([N],E)$ where $(v,u)\in E$ iff $S(v)=u$ and $P(u)=v$. The goal is to output a vertex $v\in[N]$ such that either
\begin{itemize}[noitemsep,topsep=0pt,label=$-$]
\item $v=1$ and $v$ is a non-source or a sink in $G$; or
\item $v\neq1$ and $v$ is a source or a sink in $G$.
\end{itemize}
\end{itemize}
\end{mdframed}

The prior work~\cite{BR16-CC-nash} proved an $\tOmega(N^{1/2})$ lower bound for a version of $\EoL$ where the $\Pi_v$ had communication cost $c\coloneqq \Theta(\log N)$. The cost parameter $c$ is, surprisingly, very important: later reductions (in Step 4) will incur a blow-up in input size---and hence a quantitative reduction in the eventual lower bound---that is \emph{exponential} in $c$. (Namely, when constructing payoff matrices in Step 4, the data defining a strategy for Alice will include a $c$-bit transcript of some $\Pi_v$.)

\bigskip\noindent
In this work, we obtain an optimized lower bound:
\begin{theorem}\label{thm:eol}
There is a version of \EoL with randomized communication complexity $\tOmega(N)$ where the $\Pi_v$ have constant cost $c=O(1)$.
\end{theorem}
\emph{Note:}~~Since we consider $c=O(1)$, a $c$-bit transcript of an $\Pi_v$ cannot even name arbitrary $\log(N)$-bit vertices in $[N]$. Thus we need to clarify what it means for $\Pi_v$ to ``compute'' $(S(v),P(v))$. The formal requirement is that the pair $(S(v),P(v))$ is some prescribed function of both $v$ and the $c$-bit transcript $\Pi_v(\alpha,\beta)$. Concretely, we will fix some \emph{bounded-degree host graph} $H=([N],E)$ independent of $(\alpha,\beta)$, and define graphs $G$ as subgraphs of $H$. For example, we can let $\Pi_v$ announce $S(v)$ as ``the $i$-th out-neighbor of $v$ in $H$'', which takes only $O(1)$ bits to represent.

As in \cite{BR16-CC-nash}, our lower bound is obtained by first proving an analogous result for \emph{query complexity}, and then applying a \emph{lifting theorem} that escalates the query hardness into communication hardness. A key difference is that instead of a generic lifting theorem~\cite{goos16rectangles,goos17bpp}, as used by~\cite{BR16-CC-nash}, we employ a less generic, but quantitatively better one~\cite{huynh12virtue,goos14communication}.

\paragraph{Step 1: Query lower bound.}
The query complexity analogue of $\EoL$ is defined as follows.
\begin{mdframed}
\begin{center}
\large $\QEoL_H$~~for host digraph $H=([N],E)$
\end{center}
\begin{itemize}[leftmargin=2mm]
\item {\bf Input:}~~An input $x\in\{0,1\}^E$ describes a (spanning) subgraph $G=G_x$ of $H$ consisting of the edges $e$ such that $x_e=1$.
\item {\bf Output:}~~Find a vertex $v\in [N]$ such that either
\begin{itemize}[noitemsep,topsep=0pt,label=$-$]
\item $v=1$ and $\indeg(v)\neq 0$ or $\outdeg(v)\neq 1$ in $G$; or
\item $v\neq1$ and $\indeg(v)\neq 1$ or $\outdeg(v)\neq 1$ in $G$.
\end{itemize}
\end{itemize}
\end{mdframed}

We exhibit a bounded-degree host graph $H$ such that any randomized decision tree needs to make $\tOmega(N)$ queries to the input $x$ in order to solve $\QEoL_H$. Moreover, the lower bound is proved using \emph{critical block sensitivity} (cbs), a measure introduced by Huynh and Nordstr{\"o}m~\cite{huynh12virtue} that lower bounds randomized query complexity (among other things); see \autoref{sec:def-cbs} for definitions.
\begin{restatable}{lemma}{cbslemma} \label{lem:cbs-lb}
There is a bounded-degree host graph $H=([N],E)$ such that $\cbs(\QEoL_H)\geq\tOmega(N)$.
\end{restatable}

It is not hard to prove an $\Omega(N)$ bound for a \emph{complete} host graph (equipped with successor/predecessor pointers), nor an $\Omega(N^{1/2})$ bound for a bounded-degree host graph (by reducing degrees in the complete graph via binary trees). But to achieve both an $\tOmega(N)$ bound and constant degree requires a careful choice of a host graph that has good enough routing properties. Our construction uses \emph{butterfly} graphs.

Prior to this work, a near-linear randomized query lower bound was known for a bounded-degree \emph{Tseitin} problem~\cite{goos14communication}, a canonical \PPA-complete search problem. Since $\PPAD\subseteq\PPA$, our new lower bound is qualitatively stronger (also, the proof is more involved).

\paragraph{Step 2: Communication lower bound.}
Let $R\subseteq\{0,1\}^N\times\calO$ be a query search problem (e.g., $R=\QEoL_H$), that is, on input $x\in\{0,1\}^N$ the goal is to output some $o\in\calO$ such that $(x,o)\in R$. Any such $R$ can be converted into a communication problem via \emph{gadget composition}. Namely, fix some two-party function $g\colon \Sigma\times\Sigma\to\{0,1\}$, called a \emph{gadget}. The composed search problem $R\circ g$ is defined as follows: Alice holds $\alpha\in\Sigma^N$, Bob holds $\beta\in\Sigma^N$, and their goal is to find an $o\in\calO$ such that $(x,o)\in R$ where
\[
x~\coloneqq~g^N(\alpha,\beta)~=~(g(\alpha_1,\beta_1),\ldots,g(\alpha_N,\beta_N)).
\]
It is generally conjectured that the randomized communication complexity of $R\circ g$ is characterized by the randomized query complexity of $R$, provided the gadget $g$ is chosen carefully. This was proved in~\cite{goos17bpp}, but only for a \emph{non-constant-size} gadget where Alice's input is $\Theta(\log N)$ bits. This is prohibitively large for us, since we seek protocols $\Pi_v$ of \emph{constant} communication cost. We use instead a more restricted lifting theorem due to \cite{goos14communication} (building on \cite{huynh12virtue}) that works for a constant-size gadget, but can only lift critical block sensitivity bounds.
\begin{lemma}[\cite{goos14communication}] \label{lem:lift}
There is a fixed gadget $g\colon \Sigma\times\Sigma\to\{0,1\}$ such that for any $R\subseteq\{0,1\}^N\times\calO$ the randomized communication complexity of $R\circ g$ is at least $\Omega(\cbs(R))$.
\end{lemma}

\autoref{thm:eol} now follows by combining \autoref{lem:cbs-lb} and \autoref{lem:lift}. We need only verify that the composed problem $\QEoL\circ g$ fits our \EoL template. For $v\in[N]$ consider the protocol $\Pi_v$ that computes as follows on input $(\alpha,\beta)\in\Sigma^E\times\Sigma^E$:
\begin{enumerate}[noitemsep]
\item Alice sends all symbols $\alpha_e\in\Sigma$ for $e$ incident to $v$.
\item Bob privately computes all values $x_e=g(\alpha_e,\beta_e)$ for $e$ incident to $v$.
\item Bob announces $S(v)$ as the first out-neighbor of $v$ in the subgraph determined by $x$ if such an out-neighbor exists; otherwise Bob announces $S(v)\coloneqq v$. Similarly for $P(v)$.
\end{enumerate}
This protocol has indeed cost $c=O(1)$ because $H$ is of bounded degree and $|\Sigma|$ is constant.

\subsection{Step 3: Reduction to \eps-BFP}

By Brouwer's fixed point theorem, any continuous function $f\colon [0,1]^m \rightarrow [0,1]^m$ has a fixed point, that is, $x^*$ such that $f(x^*) = x^*$. 
The BFP query problem is to find such a fixed point, given oracle access to $f$. We will consider the easier $\epsilon$-BFP problem, where we merely have to find an $x$ such $f(x)$ is $\epsilon$-close to $x$.

A theorem of~\cite{R16} reduces \QEoL~to $\epsilon$-BFP with $m = O(\log(N))$.
For our purposes, there are two downsides to using this theorem. First, it is a reduction between query complexity problems, which seems to undermine the lifting to communication we obtained in Step 2. (This obstacle was already encountered in~\cite{RW} and resolved in~\cite{BR16-CC-nash}.) 

The second issue with~\cite{R16}'s reduction is that it blows up the search space. We can discretize $[0,1]$ to obtain a finite search space. But even if the discretization used one bit per coordinate (and in fact we need a large constant number of bits), the dimension $m$ is still larger than $\log_2N$ by yet another constant factor due to the seemingly-unavoidable use of error correcting codes. All in all we have a polynomial blow-up in the size of the search space, and while that was a non-issue for~\cite{R16,BR16-CC-nash}, it is crucial for our fine-grained result.

Our approach for both obstacles is to postpone dealing with them to Step 4.
But for all the magic to happen in Step 4, we need to properly set up some infrastructure before we conclude Step~3. 
Concretely, without changing the construction of $f$ from~\cite{R16}, we observe that it can be computed in a way that is ``local'' in two different ways (we henceforth say that $f$ is {\em doubly-local}). Below is an informal description of what this means; see \autoref{subsec:Local-brouwer} for details. 
\begin{itemize}
\item First, every point $x \in [0,1]^m$ corresponds to a vertex $v$ from the host graph of the \QEoL problem%
\footnote{In fact, each $x$ corresponds to zero, one, or two vertices from the host graph, where the two vertices are either identical or neighbors. For simplicity, in this informal discussion we refer to ``the corresponding vertex''.}. We observe that in order to compute $f(x)$, one only needs local access to the neighborhood of $v$ of the \QEoL (actual, not host) graph. A similar sense of locality was used in~\cite{BR16-CC-nash}.
\item Second, if we only want to compute the $i$-th coordinate of $f(x)$, we do not even need to know the entire vector $x$. Rather, it suffices to know $x_i$, the values of $x$ on a random subset of the coordinates, and the local information of the \QEoL graph described in the previous bullet (including $v$). This is somewhat reminiscent of the local decoding used in~\cite{R16} (but our locality is much simpler and does not require any PCP machinery).
\end{itemize} 

\begin{theorem}[{\QEoL~to $\epsilon$-BFP, informal~\cite{R16}}]
There is a reduction from \QEoL over $N$ vertices to $\epsilon$-\BFP on a function $f\colon [0,1]^m \rightarrow [0,1]^m$, where $f$ is ``doubly-local''.
\end{theorem}

\subsection{Step 4: Reduction to \eps-Nash}

The existence of a Nash equilibrium is typically proved using Brouwer's fixed point theorem. McLennan and Tourky~\cite{MT05-imitation} proved the other direction, namely that the existence of a Nash equilibrium in a special {\em imitation game} implies an existence of a fixed point. Viewed as a reduction from Brouwer fixed point to Nash equilibrium, it turns out to be (roughly) approximation-preserving, and thus extremely useful in recent advances on hardness of approximation of Nash equilibrium in query complexity~\cite{B13, CCT, R16}, computational complexity~\cite{R15, R16}, and communication complexity~\cite{RW, BR16-CC-nash}.

In the basic imitation game, we think of Alice's and Bob's action space as $[0,1]^m$, and define their utility functions as follows. First, Alice chooses $x^{(\mathbf{a})} \in [0,1]^m$ that should imitate the $x^{(\mathbf{b})} \in [0,1]^m$ chosen by Bob:
$$ U^A\left(x^{(\mathbf{a})}; x^{(\mathbf{b})}\right)
 \coloneqq - \left\Vert x^{(\mathbf{a})} - x^{(\mathbf{b})}\right\Vert_2^2.$$
Notice that Alice's expected utility decomposes as 
$$ \E_{x^{(\mathbf{b})}}\left[U^A\left(x^{(\mathbf{a})}; x^{(\mathbf{b})}\right) \right]
 = - \left\Vert x^{(\mathbf{a})} - \E\left[x^{(\mathbf{b})}\right]\right\Vert_2^2 -\Var\left[x^{(\mathbf{b})}\right],$$
 where the second term does not depend on Alice's action at all. This significantly simplifies the analysis because we do not need to think about Bob's mixed strategy: in expectation, Alice just tries to get as close as possible to $\E\left[x^{(\mathbf{b})}\right]$. 
Similarly, Bob's utility function is defined as:
$$ U^B\left(x^{(\mathbf{b})}; x^{(\mathbf{a})}\right)
 \coloneqq - \left\Vert f\left(x^{(\mathbf{a})}\right) - x^{(\mathbf{b})}\right\Vert_2^2.$$
It is not hard to see that in every Nash equilibrium of the game, $x^{(\mathbf{a})} = x^{(\mathbf{b})} = f\left(x^{(\mathbf{a})}\right)$.

For our reduction, we need to make some modifications to the above imitation game.
First, observe that Bob's utility must not encode the entire function $f$---otherwise Bob could find the fixed point (or Nash equilibrium) with zero communication from Alice!
Instead, we ask that Alice's action specifies a vertex $v^{(\mathbf{a})}$, as well as her inputs to the lifting gadgets associated to (edges adjacent to) $v^{(\mathbf{a})}$. If $v^{(\mathbf{a})}$ is indeed the vertex corresponding to $x^{(\mathbf{a})}$, Bob can use his own inputs to the lifting gadgets to locally compute $f\left(x^{(\mathbf{a})}\right)$ (this corresponds to the first type of ``local''). 

The second issue is that for our fine-grained reduction, we cannot afford to let Alice's and Bob's actions specify an entire point $x \in [0,1]^m$. 
Instead, we force the equilibria of the game to be strictly mixed, where each player chooses a small (pseudo-)random subset of coordinates~$[m]$. 
Then, each player's {\em mixed strategy} represents $x \in [0,1]^m$, but each action only specifies its restriction to the corresponding subset of coordinates. 
By the second type of ``local'', Bob can locally compute the value of $f(x)$ on the intersection of subsets.
Inconveniently, the switch to mixed strategies significantly complicates the analysis: we have to make sure that Alice's mixed strategy is consistent with a single $x \in [0,1]^m$, deal with the fact that in any approximate equilibrium she is only approximately randomizing her selection of subset, etc.

Finally, the ideas above can be combined to give an $N^{1-o(1)}$ lower bound on the communication complexity (already much stronger than previous work). The bottleneck to improving further is that while we are able to distribute the vector $x$ across the support of Alice's mixed strategy, we cannot do the same with the corresponding vertex $v$ from the \EoL graph. 
The reason is that given just a single action of Alice (not her mixed strategy), Bob must be able to compute his own utility; for that he needs to locally compute $f(x)$ (on some coordinates); and even with the doubly-local property of $f$, that still requires knowing the entire $v$. 
Finally, even with the most succinct encoding, if Alice's action represents an arbitrary vertex, she needs at least $N$ actions. 
To improve to the desired $N^{2-o(1)}$ lower bound, we observe that when Bob locally computes his utility he does have another input: his own action. We thus split the encoding of $v$ between Alice's action and Bob's action, enabling us to use an \EoL host graph over $N^2$ vertices. (More generally, for an asymmetric $N^a\times N^b$ game we can split the encoding unevenly.)

\section{Critical Block Sensitivity of EoL} \label{sec:eol}

In this section, we define \emph{critical block sensitivity}, and then prove \autoref{lem:cbs-lb}, restated here:
\cbslemma*
Our construction of $H$ will have one additional property, which will be useful in \autoref{sec:game}.
\begin{fact} \label{fact:labels}
In \autoref{lem:cbs-lb}, we can take $V(H)\coloneqq \{0,1\}^n$ (where $2^n=N$) such that the labels of any two adjacent vertices differ in at most $O(1)$ coordinates.
\end{fact}

\subsection{Definitions} \label{sec:def-cbs}
Let $R\subseteq\{0,1\}^N\times\calO$ be a search problem, that is, on input $x\in\{0,1\}^N$ the goal is to output some $o\in\calO$ such that $(x,o)\in R$. We call an input $x$ \emph{critical} if it admits a unique solution. Let $f\subseteq R$ be some total function $\{0,1\}^N\to\calO$ that solves the search problem, that is, $(x,f(x))\in R$ for all $x$. The \emph{block sensitivity} of $f$ at input $x$, denoted $\bs(f,x)$, is the maximum number $b$ of pairwise-disjoint blocks $B_1,\ldots,B_b\subseteq [N]$ each of which is \emph{sensitive} for $x$, meaning $f(x)\neq f(x^{B_i})$ where $x^{B_i}$ is $x$ but with bits in $B_i$ flipped. The \emph{critical block sensitivity} of $R$~\cite{huynh12virtue} is defined as
\[
\cbs(R)~\coloneqq~\min_{f\subseteq R}\max_{\text{critical}~x}\bs(f,x).
\]

\subsection{Unbounded degree} \label{sec:unbounded-deg}

As a warm-up, we first study a simple version of the \QEoL problem relative to a \emph{complete} host graph, which is equipped with successor and predecessor \emph{pointers}. The input is a string $x\in\{0,1\}^{N'}$ of length $N'\coloneqq 2N\log N$ describing a digraph $G_x$ of in/out-degree $\leq 1$ on the vertex set~$[N]$. Specifically, for each $v \in[N]$, $x$ specifies a $\log N$-bit predecessor pointer and a $\log N$-bit successor pointer. We say there is an edge $(v,u)$ in $G_x$ iff $u$ is the successor of $v$ and $v$ is the predecessor of~$u$. The search problem is to find a vertex $v\in[N]$ such that either
\begin{itemize}[noitemsep,label=$-$]
\item $v=1$ and $v$ is a non-source or a sink in $G_x$; or
\item $v\neq1$ and $v$ is a source or a sink in $G_x$.
\end{itemize}

Let $f\colon\{0,1\}^{N'}\to[N]$ be any function that solves the search problem. Our goal is to show that there is a critical input $x$ such that $\bs(f,x)\geq \Omega(N)$.

\paragraph{Two examples.}
It is instructive to first investigate two extremal examples of $f$. What does~$f$ output on a \emph{bicritical} input consisting of two disjoint paths (one starting at the special vertex $1$)? Any function $f$ must make a choice between the unique \emph{canonical} solution (end of the path starting at vertex 1) and the other two \emph{non-canonical} solutions:
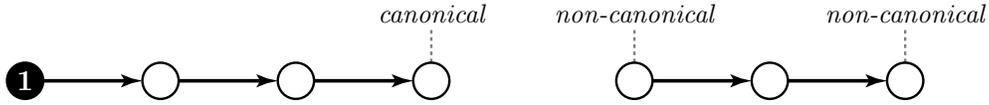
\begin{center}
\begin{lpic}[t(-5mm),b(-2mm)]{figs/bicritical(.45)}
\lbl[c]{20,20;\color{white}\boldmath$1$}
\small\itshape
\lbl[c]{140,40;canonical}
\lbl[c]{200,40;non-canonical}
\lbl[c]{280,40;non-canonical}
\end{lpic}
\end{center}

\emph{Example 1:} Suppose $f$ \emph{always} outputs the canonical solution on a bicritical input. Consider any critical input $x$ (a single path starting at vertex 1) and define a system of $N-1$ disjoint blocks $B^\del_1,\ldots,B^\del_{N-1}\subseteq[N']$ such that $x^{B^\del_i}$ is $x$ but with the $i$-th edge $(v,u)$ \emph{deleted} (say, by assigning null pointers as the successor of $v$ and as the predecessor of to $u$).
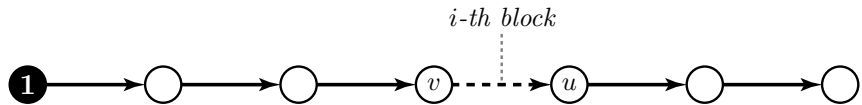
\begin{center}
\begin{lpic}[t(-4mm),b(-2mm)]{figs/bicritical-canonical(.45)}
\lbl[c]{40,20;\color{white}\boldmath$1$}
\small
\lbl[c]{160,20;$v$}
\lbl[c]{200,20;$u$}
\lbl[c]{180,40;\itshape $i$-th block}
\end{lpic}
\end{center}
Each $B^\del_i$ is sensitive for $x$ since $f(x)$ is the end of the path $x$ whereas $f(x^{B^\del_i})$ is the newly created canonical solution $v$. This shows $\bs(f,x)\geq N-1$.

\emph{Example 2:} Suppose $f$ \emph{always} outputs a non-canonical solution on a bicritical input. Consider any critical input $x$. Pair up the $i$-th and the $(N-i)$-th edge of $x$. We use these edge pairs to form $N/2-1$ blocks $B^\cut_1,\ldots,B^\cut_{N/2-1}\subseteq[N']$ (assume $N$ is even). Namely, $x^{B^\cut_i}$ is the bicritical input obtained from $x$ by \emph{shortcutting} the $i$-th pair of edges: delete the $i$-th edge pair, call them $(v,u)$ and $(v',u')$, and insert the edge $(v,u')$ (say, by assigning null pointers as the predecessor of $u$ and as the successor of $v'$, and making $(v,u')$ a predecessor--successor pair).
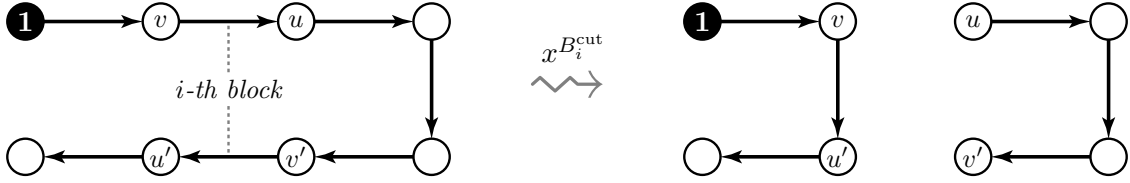
\begin{center}
\begin{lpic}[t(-2mm),b(-2mm)]{figs/bicritical-non-canonical(.45)}
\lbl[c]{20,60;\color{white}\boldmath$1$}
\lbl[c]{220,60;\color{white}\boldmath$1$}
\lbl[c]{182,52;$x^{B^\cut_i}$}
\lbl[c]{180,40;\scalebox{2.8}{\color{gray} $\leadsto$}}
\small
\lbl[c]{60,60;$v$}
\lbl[c]{100,60;$u$}
\lbl[c]{60,20;$u'$}
\lbl[c]{100,20;$v'$}
\lbl[c]{260,60;$v$}
\lbl[c]{300,60;$u$}
\lbl[c]{260,20;$u'$}
\lbl[c]{300,20;$v'$}
\lbl[c]{80,40;\itshape $i$-th block}
\end{lpic}
\end{center}
Each $B^\cut_i$ is sensitive for $x$ since $f(x)$ is the end of the path $x$ whereas $f(x^{B^\cut_i})$ is one of the newly created non-canonical solutions $u$ and $v'$. This shows $\bs(f,x)\geq N/2-1$.

\paragraph{General case.}
A general $f\colon\{0,1\}^{N'}\to[N]$ need not fall into either example case discussed above: on bicritical inputs, $f$ can decide whether or not to output a canonical solution based on the path lengths and vertex labels. However, we can still classify any $f$ according to which decision (canonical vs.\ non-canonical) it makes on \emph{most} bicritical inputs. Indeed, we define two distributions:

\begin{itemize}[itemsep=2pt]
\item \emph{Critical:} Let $\calD_1$ be the uniform distribution over critical inputs. That is, generate a directed path of length $N-1$ with vertex labels picked at random from $[N]$ (without replacement) subject to the start vertex being $1$.
\item \emph{Bicritical:} Let $\calD_2$ be a distribution over bicritical inputs generated as follows: choose an even number $\ell\in \{2,4,\ldots N-2\}$ uniformly at random, and output a graph consisting of two directed paths, one having $\ell$ vertices, the other having $N-\ell$ vertices. The vertex labels are a picked at random from~$[N]$ (without replacement) subject to the start vertex of the first path being $1$.
\end{itemize}

Given a sample $x\sim\calD_1$ we can generate a sample $y\sim\calD_2$ by either deleting (Example 1) or shortcutting (Example 2) edges of $x$. Specifically:
\begin{enumerate}[noitemsep]
\item \emph{Deletion:} Let $x\sim\calD_1$ and choose $i\in\{2,4,\ldots,N-2\}$ at random. Output $y\coloneqq x^{B^\del_i}$.
\item \emph{Shortcutting:} Let $x\sim\calD_1$ and choose $j\in[N/2-1]$ at random. Output $y\coloneqq x^{B^\cut_j}$.
\end{enumerate}
We have two cases depending on whether or not $f$ prefers canonical solutions on input $y\sim\calD_2$. Indeed, consider the probability
\[
\Pr_{y\sim\calD_2}\big[\,f(y) \text{ is canonical}\,\big]~=~
\begin{cases}\label{eq:prob}
~~\Pr_{x\sim\calD_1,i\in\{2,4,\ldots,N-2\}}\big[\,f(x^{B^\del_i}) \text{ is canonical}\,\big],\\
~~\Pr_{x\sim\calD_1,j\in[N/2-1]}\big[\,f(x^{B^\cut_j}) \text{ is canonical}\,\big].
\end{cases}
\]
\emph{Case ``\,$\geq 1/2$'':} Here $\Pr_{x,i}[\,f(x^{B^\del_i})\neq f(x)]\geq 1/2$ since $f(x)$ is non-canonical for $x^{B^\del_i}$. By averaging, there is some fixed critical input $x$ such that $\Pr_i[f(x^{B^\del_i})\neq f(x)]\geq 1/2$. But this implies $\bs(f,x)\geq\Omega(N)$, as desired.
\emph{Case ``\,$\leq 1/2$'':} Here $\Pr_{x,j}[f(x^{B^\cut_j})\neq f(x)]\geq 1/2$ since $f(x)$ is canonical for $x^{B^\cut_j}$. By averaging, there is some fixed critical input $x$ such that $\Pr_j[f(x^{B^\cut_j})\neq f(x)]\geq 1/2$. But this implies $\bs(f,x)\geq\Omega(N)$, concluding the proof (for unbounded degree).

\subsection{Logarithmic degree} \label{sec:log-degree}

Next, we prove an $\tOmega(N)$ query lower bound for $\QEoL_H$ where the host graph $H=([N],E)$ has degree $O(\log N)$. As a minor technicality, in this section, we relax the rules of the \QEoL problem (as originally defined in \autoref{sec:steps-12}) by allowing many paths to pass through a single vertex; we will un-relax this in \autoref{sec:constant-deg}. Namely, an input $x\in\{0,1\}^{E(H)}$ describes a subgraph $G_x$ of $H$ as before. The problem is to find a vertex $v\in V(H)$ such that either
\begin{itemize}[noitemsep,label=$-$]
\item $v=1$ and $\outdeg(v)\neq\indeg(v)+1$ in $G_x$; or
\item $v\neq1$ and $\outdeg(v)\neq\indeg(v)$ in $G_x$.
\end{itemize}

\paragraph{Host graph.}
For convenience, we define our host graph as a \emph{multigraph}, allowing parallel edges. We describe below a \emph{simple bounded-degree} digraph $H=([N'],E)$. The actual host graph is then taken as~$H^d$, $d\coloneqq \log N'$, defined as the graph $H$ but with each edge repeated $d$ times.

The digraph $H$ is constructed by glueing together two \emph{buttefly graphs}. The $n$-th \emph{butterfly graph} is a directed graph with $n+1$ layers, each layer containing $2^n$ vertices: the vertex set is $\{0,1\}^n\times [n+1]$ and each vertex $(z,i)$, $i\leq n$, has two out-neighbours, $(z,i+1)$ and $(z^i,i+1)$, where $z^i$ is $z$ but with the $i$-th bit flipped. Let $F_0$ and $F_1$ be two copies of the $n$-th butterfly graph. To construct~$G$, we identify the last layer of~$F_b$, $b=0,1$, with the first layer of~$F_{1-b}$. Thus $G$ has altogether $N'\coloneqq 2Nn$ vertices where $N\coloneqq 2^n$. We rename the first layer of $F_0$ (i.e., last layer of~$F_1$) as $[N]$ and the remaining vertices arbitrarily so that $V(H)=[N']$.
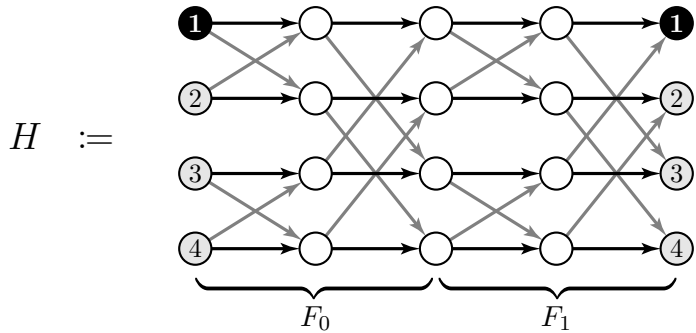
\begin{center}
\begin{lpic}[t(-3mm)]{figs/butterfly(.4)}
\Large
\lbl[r]{-7,62.5;$H~~\coloneqq$}
\normalsize
\lbl[c]{60,2;$F_0$}
\lbl[c]{140,2;$F_1$}
\small
\lbl[c]{20,100;\color{white}\boldmath$1$}
\lbl[c]{20,75;$2$}
\lbl[c]{20,50;$3$}
\lbl[c]{20,25;$4$}
\lbl[c]{180,100;\color{white}\boldmath$1$}
\lbl[c]{180,75;$2$}
\lbl[c]{180,50;$3$}
\lbl[c]{180,25;$4$}
\end{lpic}
\end{center}

\paragraph{Oblivious routing.}
To prove a critical block sensitivity lower bound we proceed analogously to the unbounded-degree proof in \autoref{sec:unbounded-deg}. The key property of $H^d$ that we will exploit is that we can embed inside $H^d$ any bounded-degree digraph $G$ on the vertex set $V(G)=[N]$. Namely, we can embed the vertices via the identity map $[N]\to[N']$, and an edge $(v,u)$ of $G$ as a $(v,u)$-path in $H^d$ (left-to-right path in the above figure) in such a way that any two edges of $G$ map to \emph{edge-disjoint} paths. Moreover, such routing can be done nearly \emph{obliviously}: each path can be chosen independently at random, and the resulting paths can be made edge-disjoint by ``local'' rearrangements. Let us make this formal.

Define $\calP_{(v,u)}$, where $(v,u)\in[N]^2$, as the uniform distribution over $(v,u)$-paths in $H$ of the minimum possible length, namely $2n$. One way to generate a path $p\sim\calP_{(v,u)}$ is to choose a vertex $w$ from $H$'s middle layer (last layer of $F_0$) uniformly at random and define $p$ as the concatenation of the \emph{unique} $(v,w)$-path in $F_0$ and the \emph{unique} $(w,u)$-path in $F_1$. Another equivalent way is to generate a random length-$n$ path starting from $v$ (in each step, choose a successor according to an unbiased coin) and then following the unique length-$n$ path from the middle layer to $u$.

Let $G=([N],E)$ be a bounded-degree digraph. We can try to embed $G$ inside $H$ by sampling a collection of paths from the product distribution $\calP_G\coloneqq \prod_{e\in E}\calP_e$. The resulting paths are likely to overlap, but not by \emph{too} much. Indeed, for an outcome $\vec{p}\in\supp\calP_G$, define the \emph{congestion} of $\vec{p}$ as the maximum over $v\in V(H)$ of the number of paths in $\vec{p}$ that touch $v$.
\begin{claim} \label{cl:congestion}
Let $G=([N],E)$ be a bounded-degree digraph. With probability $1-o(1)$ over $\vec{p}\sim\calP_G$ the congestion is $o(\log N)$.
\end{claim}
\begin{proof}
We may assume that the in/out-degree of every vertex in $G$ is at most $1$, since every bounded-degree graph is a union of constantly many such graphs. The congestion of the vertices on the last layer of $F_0$ is described by the usual balls-into-bins model, namely, $|E|\leq N$ many balls are randomly thrown into $N$ bins. It is a basic fact (e.g.,~\cite[Lemma 5.1]{mitzenmacher05probability}) that the congestion of each of these vertices is~$O(\log N/\log\log N)=o(\log N)$ with probability at least $1-1/N^2$. Similar bounds hold for vertices in the other layers, as they can be viewed as being on the last layer of some smaller butterfly graph. The claim follows by a union bound over all the vertices of $H$.
\end{proof}

Suppose $\vec{p}\in\supp\calP_G$ has congestion bounded by $d=\log N$ (which happens for $\vec{p}\sim\calP_G$ whp by \autoref{cl:congestion}). There is a natural way to embed all the paths $\vec{p}$ inside $H^d$ in an \emph{edge-disjoint} fashion. Since each edge $e$ of $H$ is used by only $\ell\leq d$ paths, say $p_1,\ldots,p_\ell$, there is room to use the $d$ parallel edges $e_1,\ldots,e_d$ corresponding to $e$ in $H^d$ to route the $p_i$ along distinct $e_i$. Such a local routing is fully specified by some injection $\pi_e\colon[\ell]\to[d]$.

We are now ready to formalize how $G$ embeds into $H^d$ via edge-disjoint paths. Namely, $G$ embeds as a distribution $\calH_G$ over subgraphs of $H^d$ (and $\bot$ for failure) defined as follows.
\begin{enumerate}[noitemsep,label=\itshape\arabic*.]
	\item Sample $\vec{p}\sim\calP_G$.
	\item If $\vec{p}$ has congestion $>d$, then output $\bot$; otherwise:
	\begin{enumerate}[noitemsep,label=\itshape\roman*.]
\item embed $\vec{p}$ \emph{randomly} into $H^d$ by choosing all the injections $(\pi_e)_{e\in E(H)}$ uniformly at random;
\item output the resulting subgraph of $H^d$.
	\end{enumerate}
\end{enumerate}
Note that $\Pr_{x\sim\calH_G}[x=\bot]\leq o(1)$ by \autoref{cl:congestion}.
\paragraph{Lower bound proof.}
Let $f\colon\{0,1\}^{E(H^d)}\to V(H^d)$ be a function that solves the \QEoL search problem relative to host graph $H^d$. Recall the distributions $\calD_1$ and $\calD_2$ over critical and bicritical graphs from \autoref{sec:unbounded-deg}. We can extend $\calD_2$ (or $\calD_1$) to a distribution over bicritical subgraphs of $H^d$ by defining $\calH_2\coloneqq\calH_{\calD_2}$, that is, $y\sim\calH_2$ is obtained by first sampling $G_2\sim\calD_2$ and then sampling $y$ from $\calH_{G_2}$. The \QEoL solutions of $y$ can be classified as canonical/non-canonical in the natural way (which respects the embedding). We again have two cases depending on whether $f$ prefers canonical solutions on input $y\sim\calH_2$. That is, consider the probability
\[
\Pr_{y\sim\calH_2}\big[\,y\neq\bot\enspace\text{and}\enspace f(y)\text{ is canonical}\,\big].
\]

\medskip\noindent
\emph{Case ``\,$\geq 1/2$'':}
Define a distribution $\calH_2'$ as follows.
\begin{enumerate}[noitemsep]
\item Sample $G_1\sim\calD_1$, and then $x\sim\calH_{G_1}$.
\item If $x=\bot$, output $y\coloneqq\bot$; otherwise:
\begin{enumerate}[noitemsep,label=\itshape\roman*.]
\item Sample an even $i\in\{2,4,\ldots,N-2\}$.
\item Output $y\coloneqq x^{B^\del_i}$, that is, $x$ but with the $i$-th path (image of $i$-th edge of $G_1$) deleted.
\end{enumerate}
\end{enumerate}
\begin{claim}\label{cl:dist}
Distributions $\calH_2$ and $\calH_2'$ are within $o(1)$ in statistical distance.
\end{claim}
\begin{proof}
Given an $y\sim\calH_2=\calH_{G_2}$ (where $G_2\sim\calD_2$ and the embedding is according to $\vec{p}\sim\calP_{G_2}$) we can generate a sample $y'\sim\calH_2'$ as follows: (1) if $y=\bot$, output $y'\coloneqq\bot$; (2) otherwise, let $e$ be the unique edge such that $G_2+e$ is critical; (3) sample $p\sim\calP_e$; (4) if $(\vec{p},p)$ does not exceed the congestion threshold $d$, output $y'\coloneqq y$ (which equals $y$ with the path $p$ embedded and then immediately deleted!); otherwise $y'\coloneqq\bot$. Here we used the \emph{oblivious routing} property: in embedding $G_2+e=G_1$ we can first embed all of $G_2$ and then the edge $e$. By \autoref{cl:congestion} we have that $y'=y$ with probability $1-o(1)$, which proves the claim.
\end{proof}

From the definition of $\calH_2'$ and \autoref{cl:dist}, we have $\Pr_{x,i}[x\neq\bot\enspace\text{and}\enspace f(x^{B^\del_i})\neq f(x)]\geq 1/2-o(1)$. By averaging, there is some fixed critical input $x$ such that $\Pr_i[f(x^{B^\del_i})\neq f(x)]\geq 1/2-o(1)$. But this implies $\cbs(f,x)\geq \Omega(N)$, as desired.

\bigskip\noindent
\emph{Case ``\,$\leq 1/2$'':}
Let $\calG$ be a distribution over graphs as illustrated below (for $N=10$) with vertex labels randomly chosen from $[N]$ subject to special vertex $1\in[N]$ being as depicted:
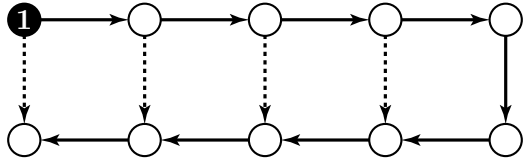
\begin{center}
\begin{lpic}[t(-3mm),b(-4mm)]{figs/embedded(.4)}
\lbl[c]{20,60;\color{white}\boldmath$1$}
\end{lpic}
\end{center}
Let $x\sim\calH_\calG$ be an embedding of a random graph from $\calG$. Assuming $x\neq\bot$, we write $\bar{x}$ for the critical input that is the subgraph of $x$ consisting of (the embeddings of) the \emph{solid} edges. Write also $\bar{x}^{B^\cut_j}$ for the bicritical input consisting of $\bar{x}$ minus its $j$-th and $(N-j)$-th edges plus the $j$-th \emph{dashed} edge. Note that $B^\cut_1,\ldots,B^\cut_{N/2-1}\subseteq E(H^d)$ is a system of pairwise-disjoint edge flips.

Define a distribution $\calH_2''$ as follows.
\begin{enumerate}[noitemsep]
\item Sample $x\sim\calH_{\calG}$.
\item If $x=\bot$, output $y\coloneqq\bot$; otherwise:
\begin{enumerate}[noitemsep,label=\itshape\roman*.]
\item Sample $j\in[N/2-1]$.
\item Output $y\coloneqq \bar{x}^{B^\cut_j}$.
\end{enumerate}
\end{enumerate}
The following claim is proved analogously to \autoref{cl:congestion}.
\begin{claim} \label{cl:dist2}
Distributions $\calH_2$ and $\calH_2''$ are within $o(1)$ in statistical distance.\qed
\end{claim}
From the definition of $\calH_2''$ and \autoref{cl:dist2}, we have $\Pr_{x,j}[x\neq\bot\enspace\text{and}\enspace f(\bar{x}^{B^\cut_j})\neq f(\bar{x})]\geq 1/2-o(1)$. By averaging, there is some fixed critical input $\bar{x}$ such that $\Pr_j[f(\bar{x}^{B^\cut_j})\neq f(\bar{x})]\geq 1/2-o(1)$. But this implies $\cbs(f,\bar{x})\geq \Omega(N)$, concluding the proof (for logarithmic degree).

\subsection{Constant degree} \label{sec:constant-deg}

\paragraph{Reducing degree.}
The digraph $H^d=([N'],E)$ has in-degree and out-degree $2d=O(\log N')$. It is easy to reduce this to a constant by replacing each vertex in $H^d$ with a bounded-degree graph $K$ that has connectivity properties similar to a complete bipartite graph between $2d$ left vertices (corresponding to incoming edges) and $2d$ right vertices (corresponding to outgoing edges). One way to construct such a graph $K$ is to start with a complete bipartite graph on $[2d]\times\{0\}\cup[2d]\times\{1\}$ and then replace each degree-$2d$ vertex with a binary tree of height $\log 2d$ (assume this is an integer). This produces a layered graph $K$ with $2\log 2d+1$ layers and $O(d^2)$ vertices. Denote by $H'=([N''],E')$ the digraph resulting from replacing each vertex of $H^d$ with a copy of $K$; formally, this construction is known as the \emph{replacement product}; see, e.g.,~\cite[\defaultS6.2]{reingold02entropy}. We have $V(H')\coloneqq V(H^d)\times V(K)$ so that $N''=\tO(N')$, which is only a polylogarithmic blow-up.

\paragraph{Lower bound (sketch).}
The critical block sensitivity lower bound in \autoref{sec:log-degree} extends naturally to the host graph $H'$. Indeed, every subgraph of $H^d$ consisting of edge-disjoint paths that we considered in \autoref{sec:log-degree} corresponds in a natural 1-to-1 way to subgraphs of $H'$ consisting of \emph{vertex}-disjoint paths. In particular, in \autoref{sec:log-degree} we allowed many paths to pass through a single vertex, but the natural mapping will now route at most one path through a vertex. We also interpret each isolated vertex in a subgraph of $H'$ as having a self-loop, so that an isolated vertex does not count as a solution to $\QEoL_{H'}$. In this way, every $f'\colon \{0,1\}^{E(H')}\to V(H')$ solving $\QEoL_{H'}$ induces an $f\colon \{0,1\}^{E(H^d)}\to V(H^d)$ solving $\QEoL_{H^d}$. This concludes the proof of \autoref{lem:cbs-lb}.

\paragraph{Vertex labels.}
Finally, we establish \autoref{fact:labels}. Namely, we argue that for $n'\coloneqq n+O(\log n)$, the vertices of $H'$ can be labeled with $n'$-bit strings having the \emph{difference property}: the labels of any two adjacent vertices differ in at most $O(1)$ coordinates. Since $V(H')=V(H)\times V(K)$ and vertices of $H'$ are adjacent only if their $H$ and $K$ parts are adjacent, it suffices to label both $H$ and $K$ appropriately and then concatenate the labels.

The vertices of $H$, viewed as $\{0,1\}^n\times[2n]$, can be made to have the difference property by just encoding the index set $[2n]$ using a Gray code. Hence it remains to label the vertices of $K$ with $O(\delta)$-bit strings for $\delta\coloneqq\log 2d$. We can view $V(K)\subseteq \{0,1\}^\delta \times \{0,1\}^\delta\times[2\delta+1]$; here, an index in $[2\delta+1]$ indicates a layer; the first layer is $\{0,1\}^\delta\times\{0\}^\delta\times\{1\}$, the last layer is $\{0\}^\delta\times\{0,1\}^\delta\times\{2\delta+1\}$. We can define adjacency similarly as in the butterfly graph so that $(v,v',i)$ is adjacent to $(u,u',j)$ only if the strings $vv'$ and $uu'$ differ in at most one position and $|i-j|\leq 1$. Moreover, we can encode the index set $[2\delta+1]$ using a Gray code.

\section{A Hard Brouwer Function} \label{sec:Brouwer}

In this section we present and slightly modify a reduction due to
\cite{R16} from \EoL (for host graph on $\{0,1\}^n$) to \BFP, the problem of finding an approximate fixed point of a continuous function $f\colon [-1,2 ]^{\Theta(n)}\to [-1,2]^{\Theta(n)}$. The reader should think of the reduction as happening between the \emph{query} variants of both problem, although we will use further properties
of the construction of $f$, as detailed in \autoref{subsec:Local-brouwer}.
The most important, and somewhat novel, part of this section is the
latter \autoref{subsec:Local-brouwer}, where we formulate
the sense in which our hard instance of \BFP is ``local''
and even ``doubly-local''.

The construction has two main components: \autoref{lem:just-a-path} shows how to embed an \EoL graph as a collection
of continuous paths in $[-1,2]^{\Theta(n)}$; \autoref{lem:HPV_2} describes how to embed a continuous Brouwer function whose fixed points correspond to endpoints of the paths constructed in \autoref{lem:just-a-path}.

\subsection{Preliminaries}

We use $\mathbf{0}_{n}$ (respectively $\mathbf{1}_{n}$) to denote
the length-$n$ vectors whose value is $0$ ($1$) in every coordinate.

\paragraph{Constants.}
This section (and the next) uses several arbitrary small constants that satisfy:
\[
0<\epsilon_{\textsc{Nash}}\ll\mbox{\ensuremath{\epsilon_{\textsc{Precision}}}}\ll\epsilon_{\textsc{Uniform}}\ll\epsilon_{\textsc{Brouwer}}\ll\delta\ll h\ll1.
\]
By this we mean that we first pick a sufficiently small constant $h$,
and then a sufficiently smaller constant $\delta$, etc. We will sometimes
use the small constants together with asymptotic notation (e.g., $O\left(\epsilon_{\textsc{Nash}}\right)$),
by which we mean ``bounded by $c\cdot\epsilon_{\textsc{Nash}}$'',
for an absolute constant $c$ (independent of $h$, $\delta$, etc.); in particular if $x=O\left(\epsilon_{\textsc{Nash}}\right)$ then $x\ll\mbox{\ensuremath{\epsilon_{\textsc{Precision}}}}$. 

Although their significance will be fully understood later, we briefly sketch their roles here:
$\epsilon_{\textsc{Nash}}$ is the approximation factor of Nash equilibrium;
$\mbox{\ensuremath{\epsilon_{\textsc{Precision}}}}$ is the precision
with which the players can specify real values (\autoref{subsec:Strategies});
$\epsilon_{\textsc{Uniform}}$ is used in the analysis in of the hard
game (\autoref{subsec:Analysis}) to bound distance-from-uniform
of certain nearly-uniform distributions; every $\epsilon_{\textsc{Nash}}$-Nash
equilibrium corresponds to an $\epsilon_{\textsc{Brouwer}}$-approximate
fixed point (\autoref{subsec:Analysis}), whereas we prove that
it is hard to find $\delta$-approximate fixed points (\autoref{sec:Brouwer});
finally, in the construction of hard Brouwer functions, $h$ quantifies
the size of special neighborhoods around special points (\autoref{sec:Brouwer}).
In \autoref{subsec:Utilities} we will define additional small
constants and relate them to the constants defined here.

\paragraph{Norms.}
We use \emph{normalized} $p$-norms: for a vector $\mathbf{x}=\left(x_{1},\dots,x_{n}\right)\in\mathbb{R}^{n}$
we define 
\[
\left\Vert \mathbf{x}\right\Vert _{p}^{p}\coloneqq\E_{i\in\left[n\right]}\left[\left(x_{i}\right)^{p}\right],
\]
where the expectation is taken wrt the uniform distribution.


\paragraph{Partitioning the coordinates.}

Let $m\coloneqq \Theta(n)$ (where the implicit constant is eventually fixed in \autoref{sec:game-prelim}). Let $\ell$ and $k$ be tiny super-constants, e.g., $\ell\coloneqq k\coloneqq\sqrt{\log n}$.
We consider two families $\sigma_{1},\dots,\sigma_{\ell^{k+1}}$ and
$\tau_{1},\dots,\tau_{\ell^{k+1}}$ of subsets of $\left[m\right]$.
Every subset has cardinality exactly $m/\ell$, and the intersection,
for every ``bichromatic'' pair of subsets satisfies $\left|\sigma_{i}\cap\tau_{j}\right|=m/\ell^{2}$.

To construct the subsets we think of the elements of $\left[m\right]$
as entries of a $\sqrt{m}\times\sqrt{m}$ matrix. Each $\sigma_{j}$
(resp.\ $\tau_{j}$) is a collection of $\sqrt{m}/\ell$ columns (resp.\
rows). Notice that this guarantees the cardinality and intersection
desiderata.

Specifically, we consider a $k$-wise independent hashing of $\left[\sqrt{m}\right]$
into $\ell$ buckets of equal size. By standard constructions (e.g., using low-degree polynomials \cite[Example 7]{Kopparty-notes}), this
can be done using $k\log\ell$ random bits. Consider all $\ell^{k+1}$
possible buckets ($\ell^{k}$ outcomes of the randomness $\times$
$\ell$ buckets for each). For $j\in\left[\ell^{k+1}\right]$, we
let $\sigma_{j}$ (resp.\ $\tau_{j}$) be the union of columns (resp.\
rows) in the $j$-th bucket. This ensures that a random $\sigma_{j}$
correspond to a $k$-wise independent subset of columns.
\paragraph{A concentration bound.} The following Chernoff-type bound for $k$-wise independent random variables is proved in~\cite[Theorem 5.I]{SSS95-chernoff-k-wise}.
\begin{theorem}[\cite{SSS95-chernoff-k-wise}]\label{thm:chernoff-k-wise}
Let $x_1,\dots, x_n\in[0,1]$ be $k$-wise independent
random variables, and let $\mu\coloneqq\E\left[\sum_{i=1}^{n}x_{i}\right]$
and $\delta\leq1$. Then
\[
\Pr\bigg[\Big|\sum_{i=1}^{n}x_{i}-\mu\Big|>\delta\mu\bigg]
~\leq~e^{-\Omega(\min\{ k,\delta^2\mu\} )}.
\]
\end{theorem}

\subsection{Embedding with a code}

Let $\eta>0$ be some sufficiently
small constant (we later set $\eta\coloneqq2\sqrt{h}$). For convenience
of notation we will construct a function $f\colon\left[-1,2\right]^{4\times m}\rightarrow\left[-1,2\right]^{4\times m}$
(instead of $\left[0,1\right]$); in particular, now the vertices
of the discrete hypercube $\left\{ 0,1\right\} ^{4\times m}$ are
interior points of our domain.%
\begin{lemma}
\label{lem:just-a-path}We can efficiently embed an \EoL
graph $G$ over $\left\{ 0,1\right\} ^{n}$ as a collection of continuous
paths and cycles in $\left[-1,2\right]^{4\times m}$, such that the
following hold:

\begin{itemize}
\item Each edge in $G$ corresponds to a concatenation of a few line segments
between vertices of $\left\{ 0,1\right\} ^{4m}$; we henceforth call
them {\em Brouwer line segments} and {\em Brouwer vertices}.
\item The points on any two non-consecutive Brouwer line segments are $\eta$-far.
\item The points on any two consecutive Brouwer line segments are also $\eta$-far,
except near the point $\mathbf{y}\in\left\{ 0,1\right\} ^{4\times m}$
where the two Brouwer line segments connect.
\item Every two consecutive Brouwer line segments are orthogonal.
\item Given any point $\mathbf{x}\in\left[-1,2\right]^{4\times m}$, we
can use the \EoL predecessor and successor oracles to
determine whether $\mathbf{x}$ is $\eta$-close to any Brouwer line
segment, and if so what is the distance to this Brouwer line segment,
and what are its endpoints.
\item There is a one-to-one correspondence between endpoints of the embedded
paths and solutions of the \EoL instance.
\end{itemize}
\end{lemma}

\begin{proof}
For point $\mathbf{x}\in\left[-1,2\right]^{4\times m}$ and index
$r\in\left\{ 1,2,3,4\right\} $, we let $\mathbf{x}_{r}$ denote the
point's $r$-th $m$-tuple of coordinates. Intuitively, each $m$-tuple
of entries of $\mathbf{x}$ represents a different piece of information:
$\mathbf{x}_{1}$ represents the current vertex in $G$, $\mathbf{x}_{2}$
the next vertex in $G$, and $\mathbf{x}_{3}$ an auxiliary compute-vs-copy bit. For now, the last $m$ coordinates are not used at all,
and serve as space fillers for their use in \autoref{lem:HPV_2}.
For $r\in\left\{ 3,4\right\} $, we further define $\mathbf{x}_{r}\coloneqq\E_{i\in\left[m\right]}\left[x_{r,i}\right]$.

Let $\Enc_{C}$$\left(\cdot\right)$ denote the encoding in a binary,
error correcting code $C$ with message length $n$, block length
$m=O\left(n\right)$, and constant relative distance (we eventually
use $C$ as in \autoref{sec:game-prelim}). We
assume wlog that $\Enc_{C}\left(\mathbf{0}_{n}\right)=\mathbf{0}_{m}$
(where $\Enc_{C}\left(\mathbf{0}_{n}\right)$ is the $C$-encoding
of the special vertex). The current and next vertex are encoded with
$C$, whereas the compute-vs-copy bit is encoded with a repetition code. 

The first Brouwer line segment goes from $\left(\mathbf{0}_{3\times m},\mathbf{2}_{m}\right)$
to $\mathbf{0}_{4\times m}$. We then add four Brouwer line segments
for each edge in the \EoL instance. Specifically, for
each edge $\left(u\rightarrow v\right)$ in $G$, we have Brouwer
line segments connecting following points (in this order): 
\begin{align*}
\mathbf{x}^{1}\left(u,v\right) & \coloneqq\left(\Enc_{C}\left(u\right),\Enc_{C}\left(u\right),\mathbf{0}_{m},\mathbf{0}_{m}\right),\\
\mathbf{x}^{2}\left(u,v\right) & \coloneqq\left(\Enc_{C}\left(u\right),\Enc_{C}\left(v\right),\mathbf{0}_{m},\mathbf{0}_{m}\right),\\
\mathbf{x}^{3}\left(u,v\right) & \coloneqq\left(\Enc_{C}\left(u\right),\Enc_{C}\left(v\right),\mathbf{1}_{m},\mathbf{0}_{m}\right),\\
\mathbf{x}^{4}\left(u,v\right) & \coloneqq\left(\Enc_{C}\left(v\right),\Enc_{C}\left(v\right),\mathbf{1}_{m},\mathbf{0}_{m}\right),\\
\mathbf{x}^{5}\left(u,v\right) & \coloneqq\left(\Enc_{C}\left(v\right),\Enc_{C}\left(v\right),\mathbf{0}_{m},\mathbf{0}_{m}\right).
\end{align*}
Notice that if $S\left(v\right)$ is the successor of $v$, then $\mathbf{x}^{5}\left(u,v\right)=\mathbf{x}^{1}\left(v,S\left(v\right)\right)$.
Notice that in each Brouwer line segment, only one subset of $m$
coordinates change. Thus whenever we are close to a line, we can successfully
decode the $3m$ fixed coordinates. Once we decode the $3m$ fixed
coordinates, we can compute what should be the values on the other
$m$ coordinates using the \EoL predecessor and successor
oracles, and determine the endpoints of the Brouwer line segment,
and then also the distance to it. Finally notice that after the first
Brouwer line segment, we have $\mathbf{x}_{4}=\mathbf{0}_{m}$ for
every point on the path.

Because at each step we update a different $m$-tuple of coordinates,
every two consecutive Brouwer line segments are orthogonal.
\end{proof}
In particular, we will come back to the following definitions:
\begin{definition}
[Brouwer vertex/segment/path]\label{def:brouwer}
A {\em Brouwer vertex} is any point of the form $\mathbf{x}^{\tau}\left(u,v\right)$, for $\tau\in\left[5\right]$ and $\left(u\rightarrow v\right)$ an
edge in the \EoL instance. A {\em Brouwer line segment} is the line segment between $\mathbf{x}^{\tau}\left(u,v\right)$
and $\mathbf{x}^{\tau+1}\left(u,v\right)$, for $\tau\in\left[4\right]$
and $u,v$ are as above. The {\em Brouwer path} is the union of all the Brouwer line segments.
\end{definition}

\subsection{Constructing a continuous function}

Let $m$ be as before, and let $0<\delta\ll h<1$ be sufficiently
small constants. The reduction from \EoL to \BFP
follows from the next lemma by setting $f\left(\mathbf{x}\right)\coloneqq\mathbf{x}+g\left(\mathbf{x}\right)$.
\begin{lemma}
\label{lem:HPV_2}We can efficiently embed an \EoL graph
$G$ over $\left\{ 0,1\right\} ^{n}$ as a displacement function $g:\left[-1,2\right]^{4\times m}\rightarrow\left[-\delta,\delta\right]^{4\times m}$
such that:

\begin{enumerate}
\item $g\left(\cdot\right)$ does not send any point outside the hypercube,
i.e., $\mathbf{x}+g\left(\mathbf{x}\right)\in\left[-1,2\right]^{4\times m}$.
\item $g\left(\cdot\right)$ is $O\left(1\right)$-Lipschitz (thus, $f\left(\cdot\right)$
is also $O\left(1\right)$-Lipschitz).
\begin{enumerate}
\item Furthermore, $g\left(\cdot\right)$ and $f\left(\cdot\right)$ also
satisfy the following stronger condition: For every coordinate $i\in\left[4\right]\times\left[m\right]$
and every $\mathbf{x},\mathbf{y}\in\left[-1,2\right]^{4\times m}$,
we have that $\left|g_{i}\left(\mathbf{x}\right)-g_{i}\left(\mathbf{y}\right)\right|\leq O\left(\max\left\{ \left|x_{i}-y_{i}\right|,\left\Vert \mathbf{x}-\mathbf{y}\right\Vert _{2}\right\} \right)$.
(Here $g_{i}\left(\cdot\right),x_{i},y_{i}$ denote the $i$-th coordinate
of $g\left(\cdot\right),\mathbf{x},\mathbf{y}$, respectively.)
\end{enumerate}
\item $\left\Vert g\left(\mathbf{x}\right)\right\Vert _{2}=\Omega\left(\delta\right)$
for every $\mathbf{x}$ that does not correspond to an endpoint of
a path.
\item The value of $g$ at any point $\mathbf{x}$ that is $2\sqrt{h}$-close
to a Brouwer line segment (resp.\ two consecutive Brouwer line segments),
depends only on its location relative to the endpoints of the Brouwer
line segment(s).
\item The value of $g$ at any point $\mathbf{x}$ that is $2\sqrt{h}$-far
from all Brouwer line segments, does not depend on the graph $G$.
\end{enumerate}
\end{lemma}

\begin{proof}
For point $\mathbf{x}\in\left[-1,2\right]^{4\times m}$, the first
three $m$-tuples are used as in \autoref{lem:just-a-path} (current,
next, and compute-vs-copy); the last $m$ coordinates represent a
special default direction in which the displacement points when far
from all Brouwer line segments. 

We say that a point $\mathbf{x}$ is in the {\em picture} if $\mathbf{x}_{4}<1/2$.
We construct $g$ separately inside and outside the picture (and make
sure that the construction agrees on the hyperplane $\mathbf{x}_{4}=1/2$).

\paragraph{Truncation.}
In order for $g\left(\cdot\right)$ to be a displacement function,
we must ensure that it never sends any points outside the hypercube,
i.e., for all $\mathbf{x}\in[-1,2]^{4\times m}$ we require that also $\mathbf{x}+g\left(\mathbf{x}\right)\in\left[-1,2\right]^{4\times m}$.
Below, it is convenient to first define an {\em untruncated} displacement
function $\hat{g}:\left[-1,2\right]^{4\times m}\rightarrow\left[-\delta,\delta\right]^{4\times m}$
which is not restricted by the above condition. We then truncate each
coordinate to fit in $\left[-1,2\right]$: $g_{i}\left(\mathbf{x}\right)=\max\left\{ -1,\min\left\{ 2,x_{i}+\hat{g}_{i}\left(\mathbf{x}\right)\right\} \right\} -x_{i}$
(where $g_{i}$ denotes the $i$-th output of $g$). It is clear that
the truncation is easy to compute and if $\hat{g}\left(\cdot\right)$
is $\left(M-1\right)$-Lipschitz, then $g\left(\cdot\right)$ is $M$-Lipschitz.
It is, however, important to make sure that the the magnitude of the
displacement is not compromised. Typically, some of the coordinates
may need to be truncated, but we design the displacement so that most
coordinates, say $99\%$, are not truncated. If $\hat{g}\left(\mathbf{x}\right)$
has a non-negligible component in at least $5\%$ of the coordinates,
then in total $g\left(\mathbf{x}\right)$ maintains a non-negligible
magnitude.

\paragraph{Inside the picture.}
We have to define the displacement for $\mathbf{x}$ far from every
line, for $\mathbf{x}$ near one line, and for $\mathbf{x}$ near
two consecutive lines (notice that by \autoref{lem:just-a-path}
no point is close to two non-consecutive lines). For $\mathbf{x}$
far from every line, we use the default displacement, which points
in the positive special direction: $\hat{g}\left(\mathbf{x}\right)\coloneqq\left(\mathbf{0}_{3\times m},\delta\cdot\mathbf{1}_{m}\right)$.
Because $\mathbf{x}$ is inside the picture, the truncated displacement
$g\left(\mathbf{x}\right)$ is close to $\hat{g}\left(\mathbf{x}\right)$,
and therefore satisfies $\left\Vert g\left(\mathbf{x}\right)\right\Vert _{2}=\Omega\left(\delta\right)$.

For $\mathbf{x}$ which is close to one line, we construct the displacement
as follows: on the line, the displacement points in the direction
of the path; at distance $h$ from the line, the displacement points
in towards the line; at distance $2h$ from the line, the displacement
points against the direction of the path; at distance $3h$, the displacement
points in the default direction. 

Formally, let $\tau_{\left(\mathbf{s}\rightarrow\mathbf{t}\right)}\left(\mathbf{x}\right)$
denote the magnitude of the component of $\mathbf{x}-\mathbf{s}$
in the direction of line $\left(\mathbf{s}\rightarrow\mathbf{t}\right)$,
\begin{equation}
\tau_{\left(\mathbf{s}\rightarrow\mathbf{t}\right)}\left(\mathbf{x}\right)\coloneqq\frac{\left(\mathbf{t}-\mathbf{s}\right)}{\left\Vert \mathbf{s}-\mathbf{t}\right\Vert _{2}^{2}}\cdot\left(\mathbf{x}-\mathbf{s}\right),\label{eq:beta-definition}
\end{equation}
where $\cdot$ denotes the (in-expectation) dot product. Let $\mathbf{z}=\mathbf{z}\left(\mathbf{x}\right)$
be the point nearest to $\mathbf{x}$ on the line; notice that $\mathbf{z}$
satisfies
\begin{equation}
\mathbf{z}=\tau_{\left(\mathbf{s}\rightarrow\mathbf{t}\right)}\left(\mathbf{x}\right)\mathbf{t}+\left(1-\tau_{\left(\mathbf{s}\rightarrow\mathbf{t}\right)}\left(\mathbf{x}\right)\right)\mathbf{s}.\label{eq:z-line}
\end{equation}
 For points near the line ($\left\Vert \mathbf{x}-\mathbf{z}\right\Vert _{2}\leq3h$),
but far from its endpoints ($\tau_{\left(\mathbf{s}\rightarrow\mathbf{t}\right)}\left(\mathbf{x}\right)\in\left[\sqrt{h},1-\sqrt{h}\right]$),
we define the displacement: 
\begin{equation}
\hat{g}\left(\mathbf{x}\right)\coloneqq\begin{cases}
\delta\frac{\left(\mathbf{t}-\mathbf{s}\right)}{\left\Vert \mathbf{t}-\mathbf{s}\right\Vert _{2}} & \left\Vert \mathbf{x}-\mathbf{z}\right\Vert _{2}=0\\
\delta\frac{\left(\mathbf{z}-\mathbf{x}\right)}{h} & \left\Vert \mathbf{x}-\mathbf{z}\right\Vert _{2}=h\\
\delta\frac{\left(\mathbf{s}-\mathbf{t}\right)}{\left\Vert \mathbf{t}-\mathbf{s}\right\Vert _{2}} & \left\Vert \mathbf{x}-\mathbf{z}\right\Vert _{2}=2h\\
\delta\left(\mathbf{0}_{3\times m},\mathbf{1}_{m}\right) & \left\Vert \mathbf{x}-\mathbf{z}\right\Vert _{2}=3h
\end{cases}\label{eq:line-displacement}
\end{equation}

At intermediate distances from the line, we interpolate: at distance
$\left\Vert \mathbf{x}-\mathbf{z}\right\Vert _{2}=\frac{1}{3}h$,
for example, we have $\hat{g}\left(\mathbf{x}\right)=\frac{2}{3}\delta\frac{\left(\mathbf{t}-\mathbf{s}\right)}{\left\Vert \mathbf{t}-\mathbf{s}\right\Vert _{2}}+\frac{1}{3}\delta\frac{\left(\mathbf{z}-\mathbf{x}\right)}{h}$.
Notice that $\left(\mathbf{t}-\mathbf{s}\right)$ is orthogonal to
both $\left(\mathbf{z}-\mathbf{x}\right)$ and $\left(\mathbf{0}_{3\times m},\mathbf{1}_{m}\right)$,
so the interpolation does not lead to cancellation. Also, every point
$\mathbf{z}$ on the line is $\Omega\left(1\right)$-far in every
coordinate from $\left\{ -1,2\right\} $, so the truncated displacement
$g\left(\mathbf{x}\right)$ still satisfies $\left\Vert g\left(\mathbf{x}\right)\right\Vert _{2}=\Omega\left(\delta\right)$.
For each case in \eqref{eq:line-displacement}, $\hat{g}\left(\cdot\right)$
is either constant, or (in the case of $\left\Vert \mathbf{x}-\mathbf{z}\right\Vert _{2}=h$)
$O\left(\delta/h\right)$-Lipschitz; by choice of $\delta\ll h$,
it follows that $\hat{g}\left(\cdot\right)$ is in particular $O\left(1\right)$-Lipschitz.
Furthermore, notice that $\left\Vert \mathbf{x}-\mathbf{z}\right\Vert _{2}$
is $1$-Lipschitz, so after interpolating for intermediate distances,
$\hat{g}\left(\cdot\right)$ continues to be $O\left(1\right)$-Lipschitz.
Notice also that at distance $3h$ the displacement defined in \eqref{eq:line-displacement}
agrees with the displacements for points far from every line, so Lipschitz
continuity is preserved.

\paragraph{Close to a vertex.}
At distance $O(\sqrt{h})$ from a Brouwer vertex (recall
$\sqrt{h}\gg h$), we use a different displacement that interpolates
between the incoming and outgoing Brouwer line segments. Consider
$\mathbf{x}$ which is close to the line from $\mathbf{s}$ to $\mathbf{y}$,
and also to the line from $\mathbf{y}$ to $\mathbf{t}$. Notice that
every two consecutive Brouwer line segments change disjoint subsets
of the coordinates, so $\left(\mathbf{s}\rightarrow\mathbf{y}\right)$
and $\left(\mathbf{y}\rightarrow\mathbf{t}\right)$ are orthogonal.
Let $\mathbf{z}_{\left(\mathbf{s}\rightarrow\mathbf{y}\right)}$ be
the point on line $\left(\mathbf{s}\rightarrow\mathbf{y}\right)$
that is at distance $\sqrt{h}$ from $\mathbf{y}$; similarly, let
$\mathbf{z}_{\left(\mathbf{y}\rightarrow\mathbf{t}\right)}$ be the
point on line $\left(\mathbf{y}\rightarrow\mathbf{t}\right)$ that
is at distance $\sqrt{h}$ from $\mathbf{y}$.

The high level idea is to ``cut the corner'' and drive the flow
along the line segment $L_{\mathbf{y}}$ that connects $\mathbf{z}_{\left(\mathbf{s}\rightarrow\mathbf{y}\right)}$
and $\mathbf{z}_{\left(\mathbf{y}\rightarrow\mathbf{t}\right)}$.
In particular, we consider points $\mathbf{x}$ that are within distance
$3h$ of $L_{\mathbf{y}}$. For all points further away (including
$\mathbf{y}$ itself), we use the default displacement. 

Our goal is to interpolate between the line displacement for $\left(\mathbf{s}\rightarrow\mathbf{y}\right)$
(which is defined up to $\tau_{\left(\mathbf{s}\rightarrow\mathbf{y}\right)}\left(\mathbf{x}\right)=1-\sqrt{h}$),
and the line displacement for $\left(\mathbf{y}\rightarrow\mathbf{t}\right)$
(which begins at $\tau_{\left(\mathbf{y}\rightarrow\mathbf{t}\right)}\left(\mathbf{x}\right)=\sqrt{h}$).
Let $\Delta_{\left(\mathbf{s}\rightarrow\mathbf{y}\right)}\left(\mathbf{x}\right)\coloneqq\tau_{\left(\mathbf{s}\rightarrow\mathbf{y}\right)}\left(\mathbf{x}\right)-\left(1-\sqrt{h}\right)$,
and $\Delta_{\left(\mathbf{y}\rightarrow\mathbf{t}\right)}\left(\mathbf{x}\right)\coloneqq\sqrt{h}-\tau_{\left(\mathbf{y}\rightarrow\mathbf{t}\right)}\left(\mathbf{x}\right)$.
We set our interpolation parameter 
\begin{equation}
\psi=\psi_{\mathbf{y}}\left(\mathbf{x}\right)\coloneqq\frac{\Delta_{\left(\mathbf{y}\rightarrow\mathbf{t}\right)}\left(\mathbf{x}\right)}{\Delta_{\left(\mathbf{y}\rightarrow\mathbf{t}\right)}\left(\mathbf{x}\right)+\Delta_{\left(\mathbf{s}\rightarrow\mathbf{y}\right)}\left(\mathbf{x}\right)}.\label{eq:psi-definition}
\end{equation}
We now define 
\begin{equation}
\mathbf{z}\coloneqq\psi\mathbf{z}_{\left(\mathbf{s}\rightarrow\mathbf{y}\right)}+\left(1-\psi\right)\mathbf{z}_{\left(\mathbf{y}\rightarrow\mathbf{t}\right)}.\label{eq:z-definition}
\end{equation}

For points $\mathbf{x}$ near $\mathbf{y}$ such that $\Delta_{\left(\mathbf{s}\rightarrow\mathbf{y}\right)}\left(\mathbf{x}\right),\Delta_{\left(\mathbf{y}\rightarrow\mathbf{t}\right)}\left(\mathbf{x}\right)\geq0$,
we can now define the displacement analogously to \eqref{eq:line-displacement}:
\begin{equation}
\hat{g}\left(\mathbf{x}\right)\coloneqq\begin{cases}
\delta\cdot\left[\psi\frac{\left(\mathbf{y}-\mathbf{s}\right)}{\left\Vert \mathbf{y}-\mathbf{s}\right\Vert _{2}}+\left(1-\psi\right)\frac{\left(\mathbf{t}-\mathbf{y}\right)}{\left\Vert \mathbf{t}-\mathbf{y}\right\Vert _{2}}\right] & \left\Vert \mathbf{x}-\mathbf{z}\right\Vert _{2}=0\\
\delta\frac{\left(\mathbf{z}-\mathbf{x}\right)}{h} & \left\Vert \mathbf{x}-\mathbf{z}\right\Vert _{2}=h\\
\delta\cdot\left[\psi\frac{\left(\mathbf{s}-\mathbf{y}\right)}{\left\Vert \mathbf{y}-\mathbf{s}\right\Vert _{2}}+\left(1-\psi\right)\frac{\left(\mathbf{y}-\mathbf{t}\right)}{\left\Vert \mathbf{t}-\mathbf{y}\right\Vert _{2}}\right] & \left\Vert \mathbf{x}-\mathbf{z}\right\Vert _{2}=2h\\
\delta\left(\mathbf{0}_{3\times m},\mathbf{1}_{m}\right) & \left\Vert \mathbf{x}-\mathbf{z}\right\Vert _{2}\geq3h
\end{cases}.\label{eq:vertex-displacement}
\end{equation}
At intermediate distances, interpolate according to $\left\Vert \mathbf{x}-\mathbf{z}\right\Vert _{2}$.
Notice that for each fixed choice of $\psi\in\left[0,1\right]$ (and
$\mathbf{z}$), $\hat{g}$ is $O\left(\delta/h\right)=O\left(1\right)$-Lipschitz.
Furthermore, $\Delta_{\left(\mathbf{s}\rightarrow\mathbf{y}\right)}$
and $\Delta_{\left(\mathbf{y}\rightarrow\mathbf{t}\right)}$ are $1$-Lipschitz
in $\mathbf{x}$. For any $\mathbf{z}\in L_{\mathbf{y}}$, $\Delta_{\left(\mathbf{y}\rightarrow\mathbf{t}\right)}\left(\mathbf{z}\right)+\Delta_{\left(\mathbf{s}\rightarrow\mathbf{y}\right)}\left(\mathbf{z}\right)=\sqrt{h}$.
For general $\mathbf{x}$, we have 
\begin{equation}
\Delta_{\left(\mathbf{y}\rightarrow\mathbf{t}\right)}\left(\mathbf{x}\right)+\Delta_{\left(\mathbf{s}\rightarrow\mathbf{y}\right)}\left(\mathbf{x}\right)\geq\Delta_{\left(\mathbf{y}\rightarrow\mathbf{t}\right)}\left(\mathbf{z}\right)+\Delta_{\left(\mathbf{s}\rightarrow\mathbf{y}\right)}\left(\mathbf{z}\right)-2\left\Vert \mathbf{x}-\mathbf{z}\right\Vert _{2}=\sqrt{h}-2\left\Vert \mathbf{x}-\mathbf{z}\right\Vert _{2};\label{eq:Delta+Delta}
\end{equation}
so $\psi$ is $O\left(1/\sqrt{h}\right)$-Lipschitz whenever $\left\Vert \mathbf{x}-\mathbf{z}\right\Vert _{2}<3h$,
and otherwise has no effect on $\hat{g}\left(\mathbf{x}\right)$.
We conclude that $\hat{g}$ is still $O\left(1\right)$-Lipschitz
when interpolating across different values of $\psi$. At the interface
with \eqref{eq:line-displacement}, $\psi$ is $1$ ($0$ near $\mathbf{z}_{\left(\mathbf{y}\rightarrow\mathbf{t}\right)}$),
so \eqref{eq:line-displacement} and \eqref{eq:vertex-displacement}
are equal. Therefore $\hat{g}$ is $O\left(1\right)$-Lipschitz on
all of $\left[-1,2\right]^{4m}$. 

To lower bound the magnitude of the displacement, we argue that $\left(\mathbf{z}-\mathbf{x}\right)$
is orthogonal to $\left[\psi\frac{\left(\mathbf{y}-\mathbf{s}\right)}{\left\Vert \mathbf{y}-\mathbf{s}\right\Vert _{2}}+\left(1-\psi\right)\frac{\left(\mathbf{t}-\mathbf{y}\right)}{\left\Vert \mathbf{t}-\mathbf{y}\right\Vert _{2}}\right]$.
First, observe that we can restrict our attention to the component
of $\left(\mathbf{z}-\mathbf{x}\right)$ that belongs to the plane
defined by $\mathbf{s},\mathbf{y},\mathbf{t}$ (in which $\mathbf{z}$
also lies). Let $P_{\mathbf{s},\mathbf{y},\mathbf{t}}\left(\mathbf{x}\right)$
denote the projection of $\mathbf{x}$ to this plain. We can write
points in this plane in terms of their $\Delta\left(\cdot\right)\coloneqq\left(\Delta_{\left(\mathbf{s}\rightarrow\mathbf{y}\right)}\left(\cdot\right),\Delta_{\left(\mathbf{y}\rightarrow\mathbf{t}\right)}\left(\cdot\right)\right)$
values. (Recall that $\left(\mathbf{s}\rightarrow\mathbf{y}\right)$
and $\left(\mathbf{y}\rightarrow\mathbf{t}\right)$ are orthogonal.) 

First, observe that $\Delta\left(\mathbf{z}_{\left(\mathbf{s}\rightarrow\mathbf{y}\right)}\right)=\left(0,\sqrt{h}\right)$,
$\Delta\left(\mathbf{z}_{\left(\mathbf{y}\rightarrow\mathbf{t}\right)}\right)=\left(\sqrt{h},0\right)$
and $\Delta\left(\mathbf{y}\right)=\left(\sqrt{h},\sqrt{h}\right)$.
Notice also that
\[
\left[\psi\frac{\left(\mathbf{y}-\mathbf{s}\right)}{\left\Vert \mathbf{y}-\mathbf{s}\right\Vert _{2}}+\left(1-\psi\right)\frac{\left(\mathbf{t}-\mathbf{y}\right)}{\left\Vert \mathbf{t}-\mathbf{y}\right\Vert _{2}}\right]=\left[\psi\frac{\left(\mathbf{y}-\mathbf{z}_{\left(\mathbf{s}\rightarrow\mathbf{y}\right)}\right)}{\sqrt{h}}+\left(1-\psi\right)\frac{\left(\mathbf{z}_{\left(\mathbf{y}\rightarrow\mathbf{t}\right)}-\mathbf{y}\right)}{\sqrt{h}}\right].
\]
Putting those together, we have that
\begin{equation}
\Delta\left(\left[\psi\frac{\mathbf{y}}{\left\Vert \mathbf{y}-\mathbf{s}\right\Vert _{2}}+\left(1-\psi\right)\frac{\mathbf{t}}{\left\Vert \mathbf{t}-\mathbf{y}\right\Vert _{2}}\right]\right)-\Delta\left(\left[\psi\frac{\mathbf{s}}{\left\Vert \mathbf{y}-\mathbf{s}\right\Vert _{2}}+\left(1-\psi\right)\frac{\mathbf{y}}{\left\Vert \mathbf{t}-\mathbf{y}\right\Vert _{2}}\right]\right)=\left(\psi,1-\psi\right).\label{eq:Delta(a(y-s)-(1-a)(t-y))}
\end{equation}

For $\mathbf{z}$, we have
\[
\Delta\left(\mathbf{z}\right)=\psi\Delta\left(\mathbf{z}_{\left(\mathbf{s}\rightarrow\mathbf{y}\right)}\right)+\left(1-\psi\right)\Delta\left(\mathbf{z}_{\left(\mathbf{y}\rightarrow\mathbf{t}\right)}\right)=\sqrt{h}\left(1-\psi,\psi\right).
\]
Finally, for $P_{\mathbf{s},\mathbf{y},\mathbf{t}}\left(\mathbf{x}\right)$,
we can write
\begin{eqnarray*}
\Delta\left(P_{\mathbf{s},\mathbf{y},\mathbf{t}}\left(\mathbf{x}\right)\right) & = & \left(\Delta_{\left(\mathbf{y}\rightarrow\mathbf{t}\right)}\left(\mathbf{x}\right),\Delta_{\left(\mathbf{s}\rightarrow\mathbf{y}\right)}\left(\mathbf{x}\right)\right)\\
 & = & \frac{1}{\Delta_{\left(\mathbf{y}\rightarrow\mathbf{t}\right)}\left(\mathbf{x}\right)+\Delta_{\left(\mathbf{s}\rightarrow\mathbf{y}\right)}\left(\mathbf{x}\right)}\left(1-\psi,\psi\right).
\end{eqnarray*}
Therefore $\Delta\left(\mathbf{z}\right)-\Delta\left(P_{\mathbf{s},\mathbf{y},\mathbf{t}}\left(\mathbf{x}\right)\right)$
is orthogonal to \eqref{eq:Delta(a(y-s)-(1-a)(t-y))}.

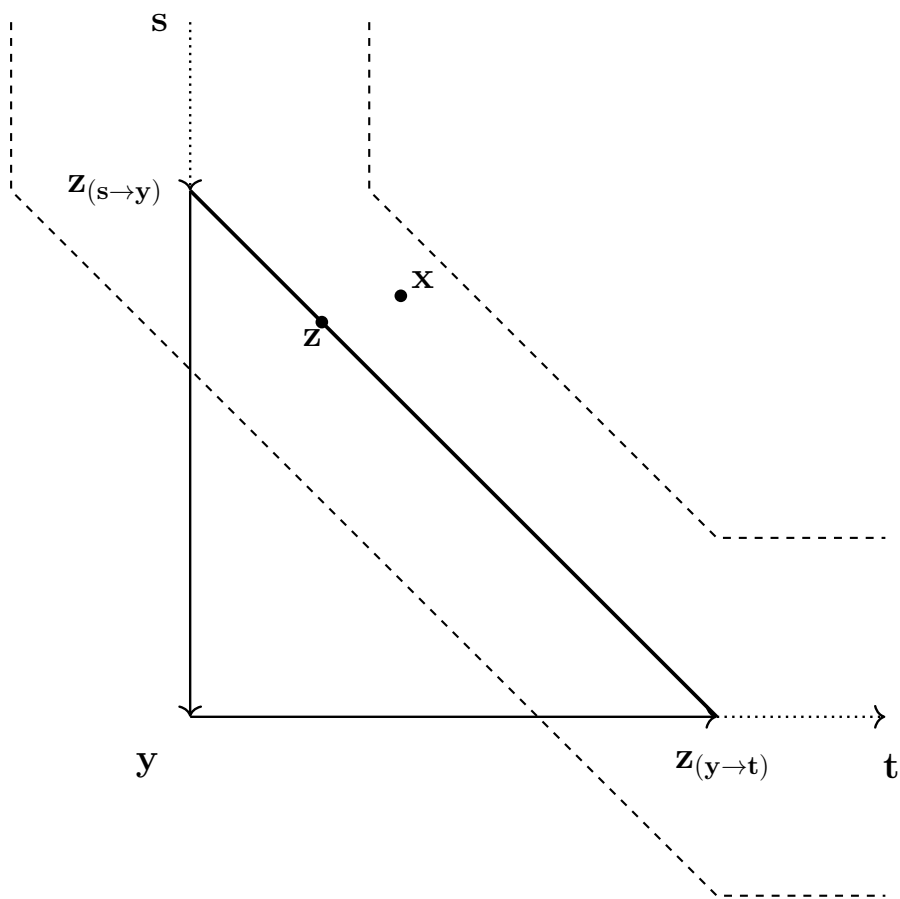
\begin{figure}
\begin{tikzpicture} [scale = 0.07] \tikzstyle{every node}=[font=\Large] \usetikzlibrary{calc}

\coordinate (s) at (0,132); \coordinate (zsy) at (0,100); \coordinate (y) at (0,0); \coordinate (zyt) at (100,0); \coordinate (t) at (132,0);
\coordinate (sPlus) at (34,132); \coordinate (sMinus) at (-34,132);
\coordinate (zsyPlus) at (34,100); \coordinate (zsyMinus) at (-34,100);
\coordinate (zytPlus) at (100,34); \coordinate (zytMinus) at (100,-34);
\coordinate (tPlus) at (132,34); \coordinate (tMinus) at (132,-34);
\coordinate (x) at (40,80); 
\coordinate (z) at (25,75);
\draw (s)+(-7,0) node {{ {$\bf{s}$}}}; \draw (zsy)+(-15.4,0) node {{ {$\bf{\mathbf{z}_{\left(\mathbf{s}\rightarrow\mathbf{y}\right)}}$}}}; \draw (y)+(-9.3,-9.3) node {{ {$\bf{y}$}}}; \draw (zyt)+(0,-9.3) node {{ {$\mathbf{z}_{\left(\mathbf{y}\rightarrow\mathbf{t}\right)}$}}}; \draw (t)+(0,-9.3) node {{ {$\bf{t}$}}};
\draw (x)+(3,3) node {{ {$\bf{x}$}}};   \fill (x)  circle[radius=1.2];
\draw (z)+(-3,-3) node {{ {$\bf{z}$}}};   \fill (z)  circle[radius=1.2];
\draw (s) -- (zsy) [dotted, -{>[scale=1.2]}, line width = 0.9]; \draw (zsy) -- (y) [-{>[scale=1.2]}, line width = 0.9]; \draw (y) -- (zyt) [-{>[scale=1.2]}, line width = 0.9]; \draw (zyt) -- (t) [dotted, -{>[scale=1.2]}, line width = 0.9];
\draw  (zsy) -- (zyt) [line width = 1.4]; \draw  (sPlus) -- (zsyPlus)-- (zytPlus) -- (tPlus) [dashed, line width = 0.8]; \draw  (sMinus) -- (zsyMinus)-- (zytMinus) -- (tMinus) [dashed, line width = 0.8];
\
\end{tikzpicture}
\vspace{1cm}
\caption{Geometry near a Brouwer vertex.
The figure (not drawn to scale) shows some of the important points
near a Brouwer vertex $\mathbf{y}$: There is an incoming Brouwer
line segment from $\mathbf{s}$ through $\mathbf{z_{\left(\mathbf{s}\rightarrow\mathbf{y}\right)}}$,
and an outgoing Brouwer line segment to $\mathbf{t}$ through $\mathbf{z}_{\left(\mathbf{y}\rightarrow\mathbf{t}\right)}$.
For each point $\mathbf{x}$ between the dashed lines, we assign a
point $\mathbf{z}$ on the line $L_{\mathbf{y}}$ as in \eqref{eq:z-definition},
and define the displacement according to \eqref{eq:vertex-displacement}.
Outside the dashed lines (including at $\mathbf{y}$ itself), we use
the default displacement $\delta\left(\mathbf{0}_{3\times m},\mathbf{1}_{m}\right)$.}
\end{figure}

\paragraph{Close to an end-of-line.}
Close the endpoint of a path, we do not have to be as careful with
defining the displacement: any Lipschitz extension of the displacement
we defined everywhere else would do, since here we are allowed (in
fact, expect) to have fixed points. 

For concreteness, let $(\mathbf{s}\rightarrow\mathbf{t})$ be the
last Brouwer line segment in a path. In \eqref{eq:line-displacement},
we defined the displacement for points $\mathbf{x}$ such that $\tau_{(\mathbf{s}\rightarrow\mathbf{t})}\left(\mathbf{x}\right)\leq1-\sqrt{h}$.
For points such that $\tau_{(\mathbf{s}\rightarrow\mathbf{t})}\left(\mathbf{x}\right)=1$
(i.e., at the hyperplane through $\mathbf{t}$ and perpendicular to
$(\mathbf{s}\rightarrow\mathbf{t})$), we simply set the default displacement
$\hat{g}(\mathbf{x})\coloneqq\delta\left(\mathbf{0}_{3\times m},\mathbf{1}_{m}\right)$.
For intermediate values of $\tau_{(\mathbf{s}\rightarrow\mathbf{t})}\left(\mathbf{x}\right)\in\left[1-\sqrt{h},h\right]$,
we simply interpolate according to $\tau_{(\mathbf{s}\rightarrow\mathbf{t})}\left(\mathbf{x}\right)$.
Notice that this induces a fixed point for some intermediate point
since for $\mathbf{x}$ directly ``above'' the Brouwer line segment,
$\delta\frac{\mathbf{z}-\mathbf{x}}{h}$ perfectly cancels $\delta\left(\mathbf{0}_{3\times m},\mathbf{1}_{m}\right)$.
Define the displacement analogously for the first Brouwer line segment
of any path (except for the Brouwer line segment from $\left(\mathbf{0}_{3\times m},2\cdot\mathbf{1}_{m}\right)$
to $\left(\mathbf{0}_{4\times m}\right)$).

\paragraph{Outside the picture.}
The displacement outside the picture is constructed by interpolating
the displacement at $\mathbf{x}_{4}=1/2$, and the displacement at
points in the ``top'' of the hypercube, where $x_{i}=2$ for every
$i$ in the last $m$ coordinates. The former displacement, where
$\mathbf{x}_{4}=1/2$ is defined to match the displacement inside
the picture. Namely, it is the default displacement everywhere except
near the first Brouwer line segment which goes ``down'' from $\mathbf{s}=\left(\mathbf{0}_{3\times m},2\cdot\mathbf{1}_{m}\right)$
to $\mathbf{t}=\left(\mathbf{0}_{4\times m}\right)$. Near this line,
it is defined according to \eqref{eq:line-displacement}. (Notice
that $\left\Vert \mathbf{t}-\mathbf{s}\right\Vert _{2}=1$.) 

Formally, let $\mathbf{z}_{1/2}=\left(\mathbf{0}_{3\times m},\frac{1}{2}\cdot\mathbf{1}_{m}\right)$;
for $\mathbf{x}$ on the boundary of the picture, we have: 
\begin{equation}
\hat{g}\left(\mathbf{x}\right)\coloneqq\begin{cases}
\delta\left(\mathbf{0}_{3\times m},-\mathbf{1}_{m}\right) & \left\Vert \mathbf{x}-\mathbf{z}_{1/2}\right\Vert _{2}=0\\
\delta\frac{\left(\mathbf{z}_{1/2}-\mathbf{x}\right)}{h} & \left\Vert \mathbf{x}-\mathbf{z}_{1/2}\right\Vert _{2}=h\\
\delta\left(\mathbf{0}_{3\times m},\mathbf{1}_{m}\right) & \left\Vert \mathbf{x}-\mathbf{z}_{1/2}\right\Vert _{2}\geq2h
\end{cases}\label{eq:picture-boundary}
\end{equation}
For points $\mathbf{x}$ such that $\mathbf{x}_{4}$ is very close
to $2$, the displacement $\delta\left(\mathbf{0}_{3\times m},\mathbf{1}_{m}\right)$
is not helpful because it points outside the hypercube, i.e., it would
get completely erased by the truncation. Instead, we define the displacement
as follows:
\begin{equation}
\hat{g}\left(\mathbf{x}\right)\coloneqq\begin{cases}
\delta\left(\mathbf{0}_{3\times m},-\mathbf{1}_{m}\right) & \left\Vert \mathbf{x}-\mathbf{z}_{2}\right\Vert _{2}=0\\
\delta\frac{\left(\mathbf{z}_{2}-\mathbf{x}\right)}{h} & \left\Vert \mathbf{x}-\mathbf{z}_{2}\right\Vert _{2}\geq h,
\end{cases}\label{eq:displacement-top}
\end{equation}
where $\mathbf{z}_{2}=\left(\mathbf{0}_{3\times m},2\cdot\mathbf{1}_{m}\right)$. 

When $\mathbf{x}_{4}\in\left(1/2,2\right)$, we interpolate between
\eqref{eq:picture-boundary} and \eqref{eq:displacement-top} according
to $\tau=\tau_{\left(\mathbf{z}_{2}\rightarrow\mathbf{z}_{1/2}\right)}\left(\mathbf{x}\right)\coloneqq\frac{2-\mathbf{x}_{4}}{3/2}$.
That is, let $\mathbf{z}$ be the vector that is $0$ on the first
three $m$-tuples, and $\mathbf{x}_{4}$ on the last $m$-tuple. Then
we define 
\begin{equation}
\hat{g}\left(\mathbf{x}\right)\coloneqq\begin{cases}
\delta\left(\mathbf{0}_{3\times m},-\mathbf{1}_{m}\right) & \left\Vert \mathbf{x}-\mathbf{z}\right\Vert _{2}=0\\
\delta\frac{\left(\mathbf{z}-\mathbf{x}\right)}{h} & \left\Vert \mathbf{x}-\mathbf{z}\right\Vert _{2}=h\\
\left[\tau\frac{\left(\mathbf{z}-\mathbf{x}\right)}{h}+\left(1-\tau\right)\left(\mathbf{0}_{3\times m},\mathbf{1}_{m}\right)\right] & \left\Vert \mathbf{x}-\mathbf{z}\right\Vert _{2}\geq2h
\end{cases}.\label{eq:displacement-outside}
\end{equation}
\end{proof}

\subsection{Doubly-local Brouwer function} \label{subsec:Local-brouwer}

In this subsection we argue that in order to compute $f_{i}\left(\mathbf{x}\right)$
(i.e., the $i$-th coordinate of $f\left(\mathbf{x}\right)$), we need
to know neither the entire function $f\left(\cdot\right)$ (which
would correspond to knowing the entire \EoL instance),
nor the entire vector $\mathbf{x}$ (which is too long to represent
with one player's action). Let $v\left(\mathbf{x}\right)$ be one
of the vertices of the \EoL graph that correspond to
the Brouwer line(s) segment near $\mathbf{x}$ (if such Brouwer line
segment exists). Instead of the entire vector $\mathbf{x}$, it suffices
to know $v\left(\mathbf{x}\right)$, and a more refined, yet still
approximate, value of $\mathbf{x}$ on a random set of coordinates,
including $x_{i}$. Instead of the entire \EoL graph,
it suffices to understand the local neighborhood $v\left(\mathbf{x}\right)$.
We say that our construction is ``doubly-locally computable'' because
we need to know neither the entire $f$ nor the entire $\mathbf{x}$.

We will now formalize the sense in which this construction is ``doubly-locally
computable''.
\begin{lemma}
\label{lem:doubly-local}There is an algorithm that outputs $f_{i}\left(\mathbf{x}\right)$,
whose inputs are:
\begin{itemize}
\item $x_{i}$
\item For each $r\in\left\{ 1,2\right\} $:
\begin{itemize}
\item if $\mathbf{x}\mid_{r}$ is $8\sqrt{h}$-close to $\Enc_{C}\left(v_{r}\right)$
for some $v_{r}\in\left\{ 0,1\right\} ^{n}$, the algorithm receives
$v_{r},S\left(v_{r}\right),P\left(v_{r}\right)$;
\item if $\mathbf{x}\mid_{r}$ is $25\sqrt{h}$-far from $\Enc_{C}\left(v\right)$
for every $v\in\left\{ 0,1\right\} ^{n}$, the algorithm receives
$\perp$;
\item if the distance to the nearest $\Enc_{C}\left(v_{r}\right)$ is in
$\left[8\sqrt{h},25\sqrt{h}\right]$, the algorithm receives either
$v_{r},S\left(v_{r}\right),P\left(v_{r}\right)$ or $\perp$.
\end{itemize}
\item For each $r\in\left\{ 1,2,3,4\right\} $, the value of $\mathbf{x}\mid_{\left\{ r\right\} \times T}$,
where $T\subset\left[m\right]$ has size $\left|T\right|=m/\ell$
and is sampled from a $k$-wise independent distribution.
\end{itemize}
For any choice of $\mathbf{x}\in\left[-1,2\right]^{4\times m}$, the
algorithm answers correctly with high probability over choice of $T$.
Furthermore, if instead of $x_{i}$ and $\mathbf{x}\mid_{\left[4\right]\times T}$
the algorithm receives $\hat{x}_{i}$ and $\hat{\mathbf{x}}\in\left[-1,2\right]^{4\times m/\ell}$
such that $\left(\hat{x}_{i}-x_{i}\right)\leq\epsilon$ and $\left\Vert \hat{\mathbf{x}}\mid_{r}-\mathbf{x}\mid_{\left\{ r\right\} \times T}\right\Vert _{2}^{2}\leq\epsilon$,
the algorithm's output is still accurate to within $\pm O\left(\epsilon\right)$
(w.h.p.). 
\end{lemma}

\begin{proof}
For $r\in\left\{ 3,4\right\} $, let $\hat{\mathbf{x}}{}_{r}$ denote
the expectation over $\hat{\mathbf{x}}\mid_{r}$. By the $k$-wise
independent Chernoff bound (\autoref{thm:chernoff-k-wise}), we
have that $\hat{\mathbf{x}}{}_{r}=\mathbf{x}{}_{r}\pm o\left(1\right)$
with probability $1-o\left(1\right)$. Assume henceforth that indeed
$\hat{\mathbf{x}}{}_{r}=\mathbf{x}{}_{r}\pm o\left(1\right)$.

\paragraph{Outside the picture.}

The algorithm uses $\hat{\mathbf{x}}_{4}$ to determine whether $\mathbf{x}$
is inside the picture. Close to the border $\mathbf{x}_{4}=1/2$ the
algorithm may confuse inside and outside the picture when $\hat{\mathbf{x}}$
is noisy, but by the stronger Lipschitz condition (see \autoref{lem:HPV_2})
on $f\left(\cdot\right)$, this confusion will only have a small effect.
Consider first $\mathbf{x}$ which is outside the picture. For $\mathbf{x}$
close to the first Brouwer line segment (from $\mathbf{z}_{2}$ to
$\left(\mathbf{0}_{4\times m}\right)$), the algorithm will receive
$v_{1}=v_{2}=\mathbf{0}_{n}$ as input. In this case, the algorithm
proceeds analogously to the case of $\mathbf{x}$ close to a line
inside the picture (described below). Otherwise, if $\mathbf{x}$
is outside the picture and far from the first Brouwer line segment
(namely, $\hat{\mathbf{x}}_{4}>1/2$ and at least one of $v_{1}\neq\mathbf{0}_{n}$
or $v_{2}\neq\mathbf{0}_{n}$), the algorithm simply applies the default
displacement. We henceforth restrict our attention to $\mathbf{x}$
inside the picture. 

\paragraph{Default displacement.}

At every point on a Brouwer line segment, at least one of $\mathbf{x}_{1}$
and $\mathbf{x}_{2}$ is equal to $\Enc_{C}\left(v\right)$ for some
$v\in\left\{ 0,1\right\} ^{n}$. Therefore, if $\mathbf{x}_{1}$ and
$\mathbf{x}_{2}$ are both far from every $\Enc_{C}\left(v\right)$,
the algorithm can simply apply the default displacement, i.e., return
$f_{i}\left(\mathbf{x}\right)=x_{i}+\delta\left(\mathbf{0}_{3\times m},\mathbf{1}_{m}\right)$.

\paragraph{Close to line.}

Whenever $\mathbf{x}$ is close to a single Brouwer line segment $\left(\mathbf{s}\rightarrow\mathbf{t}\right)$,
the algorithm can compute $\mathbf{s}$ and $\mathbf{t}$ using $v_{r},S\left(v_{r}\right),P\left(v_{r}\right)$.
Our first task is to estimate 
\begin{align*}
\tau_{\left(\mathbf{s}\rightarrow\mathbf{t}\right)}\left(\mathbf{x}\right) & =\frac{\left(\mathbf{t}-\mathbf{s}\right)}{\left\Vert \mathbf{s}-\mathbf{t}\right\Vert _{2}^{2}}\cdot\left(\mathbf{x}-\mathbf{s}\right)\\
 & =\frac{1}{\left\Vert \mathbf{s}-\mathbf{t}\right\Vert _{2}^{2}}\E_{i\in\left[4\right]\times\left[m\right]}\left[\left(t_{i}-s_{i}\right)\left(x_{i}-s_{i}\right)\right]
\end{align*}
Let 
\begin{align*}
\hat{\tau} & \coloneqq\frac{1}{\left\Vert \mathbf{s}-\mathbf{t}\right\Vert _{2}^{2}}\E_{i\in\left[4\right]\times T}\left[\left(t_{i}-s_{i}\right)\left(x_{i}-s_{i}\right)\right].
\end{align*}
By the $k$-wise independent Chernoff bound, we have that 
\begin{equation}
\hat{\tau}=\tau_{\left(\mathbf{s}\rightarrow\mathbf{t}\right)}\left(\mathbf{x}\right)\pm o\left(1\right)\label{eq:tau-hat-1}
\end{equation}
 with high probability (recall that $\left\Vert \mathbf{s}-\mathbf{t}\right\Vert _{2}^{2}=\Theta\left(1\right)$
for every $\mathbf{s},\mathbf{t}$). We assume henceforth that this
is the case.

Let $\hat{\mathbf{z}}\coloneqq\hat{\tau}\mathbf{t}+\left(1-\hat{\tau}\right)\mathbf{s}$.
By \eqref{eq:tau-hat-1}, we have that $\left\Vert \hat{\mathbf{z}}-\mathbf{z}\right\Vert _{\infty}=o\left(1\right)$.
Then, the algorithm can compute each case of \eqref{eq:line-displacement}
to within an $o\left(1\right)$ error in $\left\Vert \cdot\right\Vert _{\infty}$. 

Furthermore, by the $k$-wise Chernoff bound, we have
that $\left\Vert \hat{\mathbf{z}}\mid_{\left[4\right]\times T}-\hat{\mathbf{x}}\right\Vert _{2}=\left\Vert \mathbf{z}-\mathbf{x}\right\Vert _{2}\pm o\left(1\right)$
with high probability. Whenever this is the case, the algorithm can
also interpolate between the cases of \eqref{eq:line-displacement}
(approximately) correctly.

\paragraph{Close to a vertex.}

The displacement close to a Brouwer vertex is approximated similarly
to the displacement close to a Brouwer line segment. For $\mathbf{x}$
that is close to a Brouwer vertex $\mathbf{y}$, both $\mathbf{y}_{1}$
and $\mathbf{y}_{2}$ are equal to $\Enc_{C}\left(v_{1}\right),\Enc_{C}\left(v_{2}\right)$,
respectively, for some $v_{1},v_{2}\in\left\{ 0,1\right\} ^{n}$.
Therefore we can recover $\mathbf{y}$, as well as $\mathbf{s},\mathbf{t}$
using $v_{r},S\left(v_{r}\right),P\left(v_{r}\right)$. As in \eqref{eq:tau-hat-1}
we can (with high probability) recover $\tau_{\left(\mathbf{s}\rightarrow\mathbf{y}\right)}\left(\mathbf{x}\right)$
and $\tau_{\left(\mathbf{y}\rightarrow\mathbf{t}\right)}\left(\mathbf{x}\right)$
up to $\pm o\left(1\right)$; then we can also compute an estimate
$\hat{\psi}$ of $\psi\left(\mathbf{x}\right)$, which is accurate
to within $\pm o\left(1\right)$. The algorithm uses $\hat{\psi}$
to obtain an approximation $\hat{\mathbf{z}}\coloneqq\psi\mathbf{z}_{\left(\mathbf{s}\rightarrow\mathbf{y}\right)}+\left(1-\psi\right)\mathbf{z}_{\left(\mathbf{y}\rightarrow\mathbf{t}\right)}$,
which again satisfies $\left\Vert \hat{\mathbf{z}}-\mathbf{z}\right\Vert _{\infty}=o\left(1\right)$.
Therefore, the algorithm can compute each case of \eqref{eq:vertex-displacement}
to within an $o\left(1\right)$ error in $\left\Vert \cdot\right\Vert _{\infty}$.
Finally, applying the $k$-wise independent Chernoff bound, we have
that $\left\Vert \hat{\mathbf{z}}\mid_{\left[4\right]\times T}-\hat{\mathbf{x}}\right\Vert _{2}=\left\Vert \mathbf{z}-\mathbf{x}\right\Vert _{2}\pm o\left(1\right)$
with high probability, so the algorithm interpolates between the cases
of \eqref{eq:line-displacement} (approximately) correctly.
\end{proof}

\section{The Game} \label{sec:game}

In this section we construct an $N\times N$ game for $N=2^{\left(1+o\left(1\right)\right)n/2}$; the construction for asymmetric dimensions, $N^a\times N^b$, is similar. In particular, each player can construct their own payoff matrix from their input, $\alpha$ or $\beta$, to the \EoL problem of \autoref{thm:eol}.
We prove in \autoref{subsec:Analysis} that any approximate Nash equilibrium
of the game gives an approximate fixed point of the function $f$
from \autoref{sec:Brouwer}. Hence by \autoref{lem:HPV_2},
it further yields a solution to the \EoL instance.
And by \autoref{thm:eol}, finding the latter requires
$\tilde{\Omega}\left(2^{n}\right)=N^{2-o\left(1\right)}$ communication.

\subsection{Preliminaries} \label{sec:game-prelim}

\paragraph{Half-vertices.}
For any vertex $v=\left(v^{a}\circ v^{b}\right)\in\left\{ 0,1\right\} ^{n}$,
we call $v^{a},v^{b}$ its corresponding {\em half-vertices}. We
say $v_{1}^{a},v_{2}^{a}\in\left\{ 0,1\right\} ^{n/2}$ are {\em half-neighbors}
if there exist $v_{1}^{b},v_{2}^{b}\in\left\{ 0,1\right\} ^{n/2}$
such that $\left(v_{1}^{a}\circ v_{1}^{b}\right)$ and $\left(v_{2}^{a}\circ v_{2}^{b}\right)$
are neighbors in the \EoL host graph.

\paragraph{Error correcting codes.}

Let $C'$ be a constant rate, constant (relative) distance, error
correcting code which encodes strings in $\left\{ 0,1\right\} ^{n/2}$
as strings in $\left\{ 0,1\right\} ^{m/2}$. $C'$ induces a new code
$C$ to encode a string $v=\left(v^{a}\circ v^{b}\right)\in\left\{ 0,1\right\} ^{n}$
by 
\[
\Enc_{C}\left(v\right)\coloneqq\big(\Enc_{C'}(v^{a})\circ\Enc_{C'}(v^{b})\big).
\]
Notice that $C$ also has constant rate and (relative) distance.

\subsection{Strategies} \label{subsec:Strategies}

At equilibrium, the mixed strategies of Alice and Bob should implicitly
represent points $\mathbf{x}^{\mathcal{A}},\mathbf{x}^{\mathcal{B}}\in\left[-1,2\right]^{4\times m}$,
such that $\mathbf{x}^{\mathcal{A}}\approx\mathbf{x}^{\mathcal{B}}\approx f\left(\mathbf{x}^{{\cal A}}\right)$. 

An action $\mathbf{a}$ of Alice consists of:
\begin{itemize}
\item Two boolean vectors $v_{1}^{\left(\mathbf{a}\right)},v_{2}^{\left(\mathbf{a}\right)}\in\left\{ 0,1\right\}^{n/2}\cup\{\perp\}$.
In order to restrict the number of actions, we only allow pairs where either: (i) $v_{1}^{\left(\mathbf{a}\right)}=\perp$; (ii) $v_{2}^{\left(\mathbf{a}\right)}=\perp$; (iii) $v_{1}^{\left(\mathbf{a}\right)}=v_{2}^{\left(\mathbf{a}\right)}$; or (iv) $v_{1}^{\left(\mathbf{a}\right)},v_{2}^{\left(\mathbf{a}\right)}$ are half-neighbors. By \autoref{fact:labels}, there are only $2^{(1+o(1))n/2}$ such pairs. (Together, Alice's and Bob's vectors should represent
a pair of vertices, i.e., for each $r\in\left\{ 1,2\right\} $, $v_{r}^{\left(\mathbf{a},\mathbf{b}\right)}\coloneqq\left(v_{r}^{\left(\mathbf{a}\right)}\circ v_{r}^{\left(\mathbf{b}\right)}\right)\in V$.)
\item An index $j^{\left(\mathbf{a}\right)}\in\left[\ell^{k+1}\right]$
(this corresponds to a choice of subset of coordinates $\sigma_{j^{\left(\mathbf{a}\right)}}\subset\left[m\right]$);
\item A subset $J^{\left(\mathbf{a}\right)}\in\binom{\left[\ell^{k+1}\right]}{\ell^{k+1}/2}$.
\item A partial vector $\mathbf{x}^{\left(\mathbf{a}\right)}\in\left[-1,2\right]^{\left[4\right]\times\sigma_{j^{\left(\mathbf{a}\right)}}}$
(this partial vector should represent the restriction of $\mathbf{x}^{\left(\mathcal{A}\right)}$
to $\left[4\right]\times\sigma_{j^{\left(\mathbf{a}\right)}}$).
Moreover, the entries of $\mathbf{x}^{\left(\mathbf{a}\right)}$
are discretized to within a small additive constant $\mbox{\ensuremath{\epsilon_{\textsc{Precision}}}}$.

\item An input $\alpha^{\left(\mathbf{a}\right)}\in\left\{ 0,1\right\} ^{O\left(1\right)}$
to a constant number of ``lifting gadgets'' (specifically, the lifting
gadgets corresponding to $v_{r}^{\left(\mathbf{a},\mathbf{b}\right)}$
and its neighbors).
\end{itemize}

An action $\mathbf{b}$ of Bob is constructed almost analogously,
with two differences: (i) we drop the $\alpha$-component, because
that is Alice's input to the lifting gadget (Bob's input to the lifting
gadget will be implicit in the computation of his payoff matrix);
and (ii) we add a second partial vector $\hat{\mathbf{x}}$, which
plays a similar role as $\mathbf{x}$, but the separation will make
the proof a little easier. Formally,
\begin{align*}
\mathbf{b} & \coloneqq\left(v^{\left(\mathbf{b}\right)},j^{\left(\mathbf{b}\right)},J^{\left(\mathbf{b}\right)},\mathbf{x}^{\left(\mathbf{b}\right)},\hat{\mathbf{x}}^{\left(\mathbf{b}\right)}\right)\\
 & \in\left\{ 0,1\right\} ^{(1+o(1))n/2}\times\left[\ell^{k+1}\right]\times\binom{\left[\ell^{k+1}\right]}{\ell^{k+1}/2}\times\left[-1,2\right]^{\left[4\right]\times\tau_{j^{\left(\boldsymbol{b}\right)}}}\times\left[-1,2\right]^{\left[4\right]\times\tau_{j^{\left(\boldsymbol{b}\right)}}}.
\end{align*}

\subsection{Utilities} \label{subsec:Utilities}

We define Alice's and Bob's utilities modularly, as a sum of sub-utilities
corresponding to the different components of Alice's (respectively,
Bob's) action. In the rest of this subsection we define the $C$-utility
of Alice, for $C\in\left\{ v_{1},v_{2},\alpha_{1},\alpha_{2},j,J,\mathbf{x}\right\} $,
and denote it by $U_{C}^{A}\left(\mathbf{a};\mathbf{b}\right)$ (where
$\mathbf{a},\mathbf{b}$ denote Alice's, Bob's respective actions).
Alice's final utility is given as a sum of sub-utilities, scaled by
(constant) factors $\lambda_{C}$: 
\[
U^{A}\left(\mathbf{a};\mathbf{b}\right)\coloneqq\sum_{C}\lambda_{C}U_{C}^{A}\left(\mathbf{a};\mathbf{b}\right).
\]

Similarly for Bob, we define for $C'\in\left\{ v_{1},v_{2},j,J,\mathbf{x},\hat{\mathbf{x}}\right\}$:
\[
U^{B}\left(\mathbf{b};\mathbf{a}\right)\coloneqq\sum_{C'}\lambda_{C'}U_{C'}^{B}\left(\mathbf{b};\mathbf{a}\right).
\]

For simplicity, we use the same weights $\lambda_{C}$ for Alice
and Bob (whenever $C$ is defined for both), and also set $\lambda_{v_{1}}=\lambda_{v_{2}}$,
$\lambda_{\alpha_{1}}=\lambda_{\alpha_{2}}$, $\lambda_{j}=\lambda_{J}$,
and $\lambda_{\mathbf{x}}=\lambda_{\hat{\mathbf{x}}}$. Finally,
we choose all of them to be sufficiently small constants, and satisfy
the following constraints:
\[
\lambda_{v_{r}}/2>\epsilon_{\textsc{Nash}}+\max\left\{ \lambda_{\alpha_{r}},9\lambda_{\hat{\mathbf{x}}}\right\} 
\]
and 
\[
\epsilon_{\textsc{Brouwer}}\ll\delta^{2},h,
\]
where $\epsilon_{\textsc{Brouwer}}=\epsilon_{\textsc{Precision}}^{2}+\Theta\left(\epsilon_{\textsc{Nash}}/v_{\mathbf{x}}+\lambda_{\mathbf{x}}/\lambda_{j}\right)$.

\subsubsection*{Rounding to a vertex}

For each $r\in\left\{ 1,2\right\} $ we define
\[
U_{v_{r}}^{A}\left(v^{\left(\mathbf{a}\right)};j^{\left(\mathbf{b}\right)},\mathbf{x}^{\left(\mathbf{b}\right)}\right)\coloneqq\begin{cases}
1 & \left\Vert \Enc_{C'}\left(v_{r}^{\left(\mathbf{a}\right)}\right)\mid_{\tau_{j^{\left(\boldsymbol{b}\right)}}}-\mathbf{x}_{r}^{\left(\mathbf{b}\right)}\mid_{\tau_{j^{\left(\boldsymbol{b}\right)}}\cap\left[m/2\right]}\right\Vert _{2}<23\sqrt{h}\\
0 & v_{r}^{\left(\mathbf{a}\right)}=\perp\\
-1 & \text{otherwise}
\end{cases},
\]
and $U_{v}^{A}\coloneqq U_{v_{1}}^{A}+U_{v_{2}}^{A}$. Bob's $v$-utility,
denoted $U_{v}^{B}\left(v^{\left(\mathbf{b}\right)};j^{\left(\mathbf{a}\right)},\mathbf{x}^{\left(\mathbf{a}\right)}\right)$
is defined analogously (but over coordinates $\sigma_{j^{\left(\mathbf{a}\right)}}\setminus\left[m/2\right]$).

\subsubsection*{Implementing the lifting gadget}

For each $r\in\left\{ 1,2\right\} $ given $v^{\left(\mathbf{a}\right)},v^{\left(\mathbf{b}\right)}$,
let $v_{r}^{\left(\mathbf{a},\mathbf{b}\right)}\coloneqq\left(v_{r}^{\left(\mathbf{a}\right)}\circ v_{r}^{\left(\mathbf{b}\right)}\right)\in V\cup\left\{ \perp\right\} $
denote the vertex whose description corresponds to the concatenation
of $v_{r}^{\left(\mathbf{a}\right)},v_{r}^{\left(\mathbf{b}\right)}$
(where concatenation of $\perp$ and anything is still $\perp$).
Let $\alpha\left(v_{r}^{\left(\mathbf{a},\mathbf{b}\right)}\right)$
denote Alice's input to the lifting gadget corresponding to $v_{r}^{\left(\mathbf{a},\mathbf{b}\right)}$
(where $\alpha\left(\perp\right)\coloneqq\perp$). Her gadget utility
is defined as 
\[
U_{\alpha_{r}}^{A}\left(\alpha^{\left(\mathbf{a}\right)},v^{\left(\mathbf{a}\right)};v^{\left(\mathbf{b}\right)}\right)\coloneqq\begin{cases}
1 & \alpha_{r}^{\left(\mathbf{a}\right)}=\alpha\left(v_{r}^{\left(\mathbf{a},\mathbf{b}\right)}\right)\\
0 & \text{otherwise}
\end{cases},
\]
and $U_{\alpha}^{A}\coloneqq U_{\alpha_{1}}^{A}+U_{\alpha_{2}}^{A}$.

We also define $\beta\left(v_{r}^{\left(\mathbf{a},\mathbf{b}\right)}\right)$
as Bob's input to the lifting gadgets. But we do not need to explicitly
make it part of Bob's strategy because Alice doesn't need to know
the local structure of $f\left(\cdot\right)$ in order to compute
her utility.

\subsubsection*{Locally computing the Brouwer function}

For $r\in\left\{ 1,2\right\} $ let $P^{\left(\mathbf{a}\right)}\left(v_{r}^{\left(\mathbf{a},\mathbf{b}\right)}\right),S^{\left(\mathbf{a}\right)}\left(v_{r}^{\left(\mathbf{a},\mathbf{b}\right)}\right)$
be the predecessor and successor of $v_{r}^{\left(\mathbf{a},\mathbf{b}\right)}$,
as computed by Bob according to $\alpha_{r}^{\left(\mathbf{a}\right)}$
and $\beta\left(v_{r}^{\left(\mathbf{a},\mathbf{b}\right)}\right)$.
If either $v_{r}^{\left(\mathbf{a},\mathbf{b}\right)}=\perp$ or $\alpha^{\left(\mathbf{a}\right)}=\perp$,
we define $P^{\left(\mathbf{a}\right)}\left(v_{r}^{\left(\mathbf{a},\mathbf{b}\right)}\right)=S^{\left(\mathbf{a}\right)}\left(v_{r}^{\left(\mathbf{a},\mathbf{b}\right)}\right)=\perp$.
(Notice that only Bob can compute those since Alice does not know
$\beta\left(v_{r}^{\left(\mathbf{a},\mathbf{b}\right)}\right)$.)

For $i\in\sigma_{j^{\left(\mathbf{a}\right)}}\cap\tau_{j^{\left(\mathbf{b}\right)}}$,
let $f_{i}^{\left(\mathbf{a},\mathbf{b}\right)}\left(\mathbf{x}^{\left(\mathbf{a}\right)}\right)$
be the output of the doubly-local algorithm from \autoref{subsec:Local-brouwer},
given inputs $v_{r}^{\left(\mathbf{a},\mathbf{b}\right)},P^{\left(\mathbf{a}\right)}\left(v_{r}^{\left(\mathbf{a},\mathbf{b}\right)}\right),S^{\left(\mathbf{a}\right)}\left(v_{r}^{\left(\mathbf{a},\mathbf{b}\right)}\right)$
and $\mathbf{x}^{\left(\mathbf{a}\right)}$. Let $f^{\left(\mathbf{a},\mathbf{b}\right)}\left(\mathbf{x}^{\left(\mathbf{a}\right)}\right)$
denote the vector that has $f_{i}^{\left(\mathbf{a},\mathbf{b}\right)}\left(\mathbf{x}^{\left(\mathbf{a}\right)}\right)$
for each $i\in\sigma_{j^{\left(\mathbf{a}\right)}}\cap\tau_{j^{\left(\mathbf{b}\right)}}$.

We set the $\hat{\mathbf{x}}$-utility of Bob as
\[
U_{\hat{\mathbf{x}}}^{B}\left(\hat{\mathbf{x}}^{\left(\mathbf{b}\right)},j^{\left(\mathbf{b}\right)},v^{\left(\mathbf{b}\right)};\alpha^{\left(\mathbf{a}\right)},j^{\left(\mathbf{a}\right)},v^{\left(\mathbf{a}\right)},\mathbf{x}^{\left(\mathbf{a}\right)}\right)\coloneqq-\left\Vert \hat{\mathbf{x}}^{\left(\mathbf{b}\right)}\mid_{\sigma_{j^{\left(\mathbf{a}\right)}}\cap\tau_{j^{\left(\mathbf{b}\right)}}}-f^{\left(\mathbf{a},\mathbf{b}\right)}\left(\mathbf{x}^{\left(\mathbf{a}\right)}\right)\right\Vert _{2}^{2}.
\]
Alice's $\mathbf{x}$-utility is simpler, since it does not depend
on the Brouwer function (or the communication problem) at all:
\[
U_{\mathbf{x}}^{A}\left(\mathbf{x}^{\left(\mathbf{a}\right)},j^{\left(\mathbf{a}\right)};\hat{\mathbf{x}}^{\left(\mathbf{b}\right)},j^{\left(\mathbf{b}\right)}\right)\coloneqq-\left\Vert \hat{\mathbf{x}}^{\left(\mathbf{b}\right)}\mid_{\sigma_{j^{\left(\mathbf{a}\right)}}\cap\tau_{j^{\left(\mathbf{b}\right)}}}-\mathbf{x}^{\left(\mathbf{a}\right)}\mid_{\sigma_{j^{\left(\mathbf{a}\right)}}\cap\tau_{j^{\left(\mathbf{b}\right)}}}\right\Vert _{2}^{2}.
\]
Finally, we think of $\mathbf{x}^{\left(\mathbf{b}\right)}$ as a
stabler version of $\hat{\mathbf{x}}^{\left(\mathbf{b}\right)}$;
e.g., it will be easy to prove that in every approximate Nash equilibrium
Bob's choice of $\mathbf{x}^{\left(\mathbf{b}\right)}$ is essentially
deterministic (or ``pure''). Bob chooses $\mathbf{x}^{\left(\mathbf{b}\right)}$
simply to imitate Alice's $\mathbf{x}^{\left(\mathbf{a}\right)}$:
\[
U_{\mathbf{x}}^{B}\left(\mathbf{x}^{\left(\mathbf{b}\right)},j^{\left(\mathbf{b}\right)};\mathbf{x}^{\left(\mathbf{a}\right)},j^{\left(\mathbf{a}\right)}\right)\coloneqq-\left\Vert \mathbf{x}^{\left(\mathbf{b}\right)}\mid_{\sigma_{j^{\left(\mathbf{a}\right)}}\cap\tau_{j^{\left(\mathbf{b}\right)}}}-\mathbf{x}^{\left(\mathbf{a}\right)}\mid_{\sigma_{j^{\left(\mathbf{a}\right)}}\cap\tau_{j^{\left(\mathbf{b}\right)}}}\right\Vert _{2}^{2}.
\]

\subsubsection*{Enforcing near-uniform distribution over $i,j$}

For each of $j^{\left(\mathbf{a}\right)},j^{\left(\mathbf{b}\right)}$,
Alice and Bob play a generalized hide-and-seek win-lose zero-sum game
due to Althofer \cite{Althofer1993}. For $j^{\left(\mathbf{a}\right)}$,
for example, Bob is guessing a subset $J^{\left(\mathbf{b}\right)}$
whose aim is to catch Alice's $j^{\left(\mathbf{a}\right)}$. Formally,
we define: 
\begin{align*}
U_{J}^{B}\left(J^{\left(\mathbf{b}\right)};j^{\left(\mathbf{a}\right)}\right) & \coloneqq\begin{cases}
1 & j^{\left(\mathbf{a}\right)}\in J^{\left(\mathbf{b}\right)}\\
-1 & j^{\left(\mathbf{a}\right)}\notin J^{\left(\mathbf{b}\right)}
\end{cases}\\
U_{j}^{A}\left(j^{\left(\mathbf{a}\right)};J^{\left(\mathbf{b}\right)}\right) & \coloneqq-U_{J}^{B}\left(J^{\left(\mathbf{b}\right)};j^{\left(\mathbf{a}\right)}\right).
\end{align*}
The $j^{\left(\mathbf{b}\right)}$ game is defined analogously.

\subsection{Analysis} \label{subsec:Analysis}

It remains to show, for a sufficiently small constant $\epsilon>0$, that every $\epsilon$-approximate Nash equilibrium ($\epsilon$-ANE) for our game implies an approximate fixed point to $f$. For convenience, it is enough to show this for \emph{well-supported} Nash equilibria: We say that $(\calA,\calB)$ is an \emph{$\epsilon$-Well-Supported Nash equilibrium} ($\epsilon$-WSNE) if every $a$ in the support of $\calA$ is $\epsilon$-optimal against Bob's mixed strategy $\calB$, and the same holds with roles reversed. Formally, the condition for Alice is
\[
\forall a\in\supp\calA\colon\quad
\E_{b\sim\calB}\left[U^A(a,b)\right]~\geq~
\max_{a'\in S_A} \enspace\E_{b\sim\calB}\left[U^A(a',b)\right]-\epsilon.
\]
Indeed, as observed by \cite{DGP}, every $\epsilon$-ANE can be pruned to an $O(\sqrt{\epsilon})$-WSNE:
\begin{lemma}
[{\cite[Lemma 15]{DGP}}]\label{lem:ANE-WSNE}
Given an $\epsilon$-ANE for a two-player game, an $O(\sqrt{\epsilon})$-WSNE
can be obtained by removing, from each player's mixed strategy, all
actions that are not $(\epsilon+\sqrt{\epsilon})$-optimal
(with respect to the other player's mixed strategies in the original
$\epsilon$-ANE).
\end{lemma}


Given mixed strategy ${\cal A}$, let $\mathbf{x}^{{\cal A}}\in\left[-1,2\right]^{4\times m}$
denote the coordinate-wise expectation over $\mathbf{x}^{\left(\mathbf{a}\right)}$: 
\[
\left[\mathbf{x}^{{\cal A}}\right]_{i}\coloneqq\E_{\mathbf{a}\sim{\cal A}}\left[\left[\mathbf{x}^{\left(\mathbf{a}\right)}\right]_{i}\mid i\in\sigma_{j^{\left(\mathbf{a}\right)}}\right].
\]
If $i\notin\sigma_{j^{\left(\mathbf{a}\right)}}$ for all $\mathbf{a}\in\supp\left({\cal A}\right)$,
define $\left[\mathbf{x}^{{\cal A}}\right]_{i}$ arbitrarily (by \autoref{claim:nash-uniform_i,j} in an approximate Nash equilibrium this
can only happen for a negligible fraction of coordinates). For $r\in\left\{ 1,2,3,4\right\} $
let $\mathbf{x}_{r}^{{\cal A}}\coloneqq\mathbf{x}^{{\cal A}}\mid_{\left\{ r\right\} \times\left[m\right]}$
denote the restriction of $\mathbf{x}^{{\cal A}}$ to $r$-th $m$-tuple
of coordinates. Define $\mathbf{x}^{{\cal B}},\hat{\mathbf{x}}^{{\cal B}}\in\left[-1,2\right]^{4\times m}$
analogously.
\begin{proposition}
\label{prop:nash-A-f(A)}Let $\left({\cal A},{\cal B}\right)$ be
an $\epsilon_{\textsc{Nash}}$-WSNE of the game; then $\left\Vert \hat{\mathbf{x}}^{{\cal B}}-f\left(\hat{\mathbf{x}}^{{\cal B}}\right)\right\Vert _{2}^{2}=O\left(\epsilon_{\textsc{Brouwer}}\right)$.
\end{proposition}

The rest of this section completes the proof of \autoref{prop:nash-A-f(A)}.
As described in the beginning of this section, \autoref{prop:nash-A-f(A)},
together with \autoref{lem:HPV_2}, \autoref{thm:eol},
and \autoref{lem:ANE-WSNE}, imply our main result (\autoref{thm:main}).

\subsubsection*{Enforcing near-uniform distribution over $i,j$}
\begin{lemma}
[Lemma 3 in the full version of \cite{Daskalakis-Papadimitriou-PTAS}]\label{lem:DP-lemma}
Let $\left\{ a_{i}\right\} _{i=1}^{n}$ be real numbers satisfying
the following properties for some $\theta>0$: (1) $a_{1}\geq a_{2}\geq\dots\geq a_{n}$;
(2) $\sum a_{i}=0$; (3) $\sum_{i=1}^{n/2}a_{i}\leq\theta$. Then
$\sum_{i=1}^{n}\left|a_{i}\right|\leq4\theta$.
\end{lemma}

\begin{claim}
\label{claim:nash-uniform_i,j}There exists a constant $\epsilon_{\textsc{Uniform}}=O\left(\frac{\epsilon_{\textsc{Nash}}+\lambda_{\mathbf{x}}}{\lambda_{j}}\right)$,
such that for $\mathbf{a}\sim{\cal A},\mathbf{b}\sim{\cal B}$, the
marginal distributions on $j^{\left(\mathbf{a}\right)},j^{\left(\mathbf{b}\right)}$
are $\epsilon_{\textsc{Uniform}}$-close to uniform in total variation
distance.
\end{claim}

\begin{proof}
In her $j$-utility, Alice can guarantee an expected payoff of $0$
by randomizing uniformly over her choice of $j^{\left(\mathbf{a}\right)}$.
By \autoref{lem:DP-lemma}, if Alice's marginal distribution over
the choice of $j^{\left(\mathbf{a}\right)}$ is $\left(8\epsilon_{\textsc{Nash}}+36\lambda_{\mathbf{x}}\right)/\lambda_{J}$-far
from uniform (in total variation distance), then Bob can guess that
$j^{\left(\mathbf{a}\right)}$ is in some subset $J^{\left(\mathbf{b}\right)}\in\binom{\left[\ell^{k+1}\right]}{\ell^{k+1}/2}$
with advantage (over guessing at random) of at least $\left(2\epsilon_{\textsc{Nash}}+9\lambda_{\mathbf{x}}\right)/\lambda_{J}$.
This means that there exists a choice of $J^{\left(\mathbf{b}\right)}$
such that Bob can guarantee himself 
\[
\E_{\mathbf{a}}\left[U_{J}^{B}\left(J^{\left(\mathbf{b}\right)};j^{\left(\mathbf{a}\right)}\right)\right]\geq\left(2\epsilon_{\textsc{Nash}}+9\lambda_{\mathbf{x}}\right)/\lambda_{J}.
\]

Since $J^{\left(\mathbf{b}\right)}$ does not affect any other portion
of Bob's strategy, in any $\epsilon_{\textsc{Nash}}$-WSNE, and for
every $\mathbf{b}$ in Bob's support, 
\[
\E_{\mathbf{a}}\left[U_{J}^{B}\left(J^{\left(\mathbf{b}\right)};j^{\left(\mathbf{a}\right)}\right)\right]\geq\left(\epsilon_{\textsc{Nash}}+9\lambda_{\mathbf{x}}\right)/\lambda_{J}.
\]
Therefore since Althofer's gadget is a zero-sum game, we have 
\[
\E_{\mathbf{a}}\left[U_{j}^{A}\left(j^{\left(\mathbf{a}\right)};J^{\left(\mathbf{b}\right)}\right)\right]\leq-\left(\epsilon_{\textsc{Nash}}+9\lambda_{\mathbf{x}}\right)/\lambda_{J}=-\left(\epsilon_{\textsc{Nash}}+9\lambda_{\mathbf{x}}\right)/\lambda_{j}.
\]
I.e., Alice has a deviation (to uniform) that improves her $j$-utility
by $\left(\epsilon_{\textsc{Nash}}+9\lambda_{\mathbf{x}}\right)/\lambda_{j}$.
$j^{\left(\mathbf{a}\right)}$ also affects Alice's $\mathbf{x}$-utility,
but its contribution to the total utility is at most $9\lambda_{\mathbf{x}}$.
Therefore, if $j^{\left(\mathbf{a}\right)}$ is $\left(8\epsilon_{\textsc{Nash}}+4\lambda_{\mathbf{x}}\right)$-far
from uniform, then Alice has an $\epsilon_{\textsc{Nash}}$-improving
deviation, contradicting the premise. Analogous arguments hold for
$j^{\left(\mathbf{b}\right)}$.
\end{proof}
\begin{definition}
Let $\overline{{\cal A}}$ denote a mixed strategy where $\overline{{\cal A}}$
induces a uniform marginal distribution over choice of $j^{\left(\mathbf{a}\right)}$,
and $\overline{{\cal A}}$ is $\epsilon_{\textsc{Uniform}}$-close
to ${\cal A}$ (among all such distributions that are close to ${\cal A}$,
fix one arbitrarily). Define $\overline{{\cal B}}$ analogously.
\end{definition}

The following is a corollary of \autoref{claim:nash-uniform_i,j}
\begin{claim}
\label{claim:=00005Cbar=00007BA=00007D-A}$\left\Vert \mathbf{x}^{\overline{{\cal A}}}-\mathbf{x}^{{\cal A}}\right\Vert _{2}^{2},\left\Vert \mathbf{x}^{\overline{{\cal B}}}-\mathbf{x}^{{\cal B}}\right\Vert _{2}^{2},\left\Vert \hat{\mathbf{x}}^{\overline{{\cal B}}}-\hat{\mathbf{x}}^{{\cal B}}\right\Vert _{2}^{2}\leq9\epsilon_{\textsc{Uniform}}.$
\end{claim}

\subsubsection*{Rounding to a vertex}
\begin{claim}
\label{claim:pure}There exists a constant 
\[
\epsilon_{\textsc{Brouwer}}=\epsilon_{\textsc{Precision}}^{2}+\epsilon_{\textsc{Nash}}/\lambda_{\mathbf{x}}+O\left(\epsilon_{\textsc{Uniform}}\right)=\epsilon_{\textsc{Precision}}^{2}+O\left(\epsilon_{\textsc{Nash}}/\lambda_{\mathbf{x}}+\lambda_{\mathbf{x}}/\lambda_{j}\right),
\]
such that for every $\mathbf{a}\in\supp\left({\cal A}\right)$, $\left\Vert \mathbf{x}^{\left(\mathbf{a}\right)}-\mathbf{\hat{x}}^{\overline{{\cal B}}}\mid_{\sigma_{j^{\left(\mathbf{a}\right)}}}\right\Vert _{2}^{2}\leq\epsilon_{\textsc{Brouwer}}$,
and for every $\mathbf{b}\in\supp\left({\cal B}\right)$, $\left\Vert \mathbf{x}^{\left(\mathbf{b}\right)}-\mathbf{x}^{\overline{{\cal A}}}\mid_{\tau_{j^{\left(\mathbf{b}\right)}}}\right\Vert _{2}^{2}\leq\epsilon_{\textsc{Brouwer}}$.
\end{claim}

\begin{proof}
Fix an action $\mathbf{a}$ for Alice, and fix some $i\in\sigma_{j^{\left(\mathbf{a}\right)}}$.
Then Alice's expected $\mathbf{x}$-utility when playing action $\mathbf{a}$
against $\overline{{\cal B}}$ decomposes as:

\begin{align*}
\E_{\mathbf{b}\sim\overline{{\cal B}}}\left[U_{\mathbf{x}}^{A}\left(\mathbf{a};\mathbf{b}\right)\right] & =-\E_{\mathbf{b}\sim\overline{{\cal B}}}\left[\left\Vert \hat{\mathbf{x}}^{\left(\mathbf{b}\right)}\mid_{\sigma_{j^{\left(\mathbf{a}\right)}}\cap\tau_{j^{\left(\mathbf{b}\right)}}}-\mathbf{x}^{\left(\mathbf{a}\right)}\mid_{\sigma_{j^{\left(\mathbf{a}\right)}}\cap\tau_{j^{\left(\mathbf{b}\right)}}}\right\Vert _{2}^{2}\right].
\end{align*}
Let $\overline{{\cal B}}\left(i\right)$ denote the restriction of
$\overline{{\cal B}}$ to actions $\mathbf{b}$ such that $i\in\tau_{j}\left(\mathbf{b}\right)$.
We now have that, 
\begin{align}
\E_{\mathbf{b}\sim\overline{{\cal B}}}\left[U_{\mathbf{x}}^{A}\left(\mathbf{a};\mathbf{b}\right)\right] & =-\E_{\mathbf{b}\sim\overline{{\cal B}}}\left[\E_{i\in\sigma_{j^{\left(\mathbf{a}\right)}}\cap\tau_{j^{\left(\mathbf{b}\right)}}}\left[\left(\hat{\mathbf{x}}^{\left(\mathbf{b}\right)}\mid_{i}-\mathbf{x}^{\left(\mathbf{a}\right)}\mid_{i}\right)^{2}\right]\right]\label{eq:=00005Csigma=00005Ccap=00005Ctau}\\
 & =-\E_{i\in\sigma_{j^{\left(\mathbf{a}\right)}}}\left[\E_{\mathbf{b}\sim\overline{{\cal B}}\left(i\right)}\left[\left(\hat{\mathbf{x}}^{\left(\mathbf{b}\right)}\mid_{i}-\mathbf{x}^{\left(\mathbf{a}\right)}\mid_{i}\right)^{2}\right]\right]\label{eq:sigmaX=00005Ctau}\\
 & =-\E_{i\in\sigma_{j^{\left(\mathbf{a}\right)}}}\left[\underbrace{\Var_{\mathbf{b}\sim\overline{{\cal B}}\left(i\right)}\left(\hat{\mathbf{x}}^{\left(\mathbf{b}\right)}\mid_{i}\right)}_{\coloneqq V}+\left(\hat{\mathbf{x}}^{\overline{{\cal B}}}\mid_{i}-\mathbf{x}^{\left(\mathbf{a}\right)}\mid_{i}\right)^{2}\right].\label{eq:x^A-decompose}
\end{align}
Here the transition from \eqref{eq:=00005Csigma=00005Ccap=00005Ctau}
to \eqref{eq:sigmaX=00005Ctau} follows because $\overline{{\cal B}}$
has uniform marginals over $j^{\left(\mathbf{b}\right)}$ and $\left|\sigma_{j^{\left(\mathbf{a}\right)}}\cap\tau_{j^{\left(\mathbf{b}\right)}}\right|=m/\ell^{2}$
for any choices of $j^{\left(\mathbf{a}\right)},j^{\left(\mathbf{b}\right)}$.
Notice that the variance $V$ in the last equation does not depend
at all on Alice's choice of $\mathbf{x}^{\left(\mathbf{a}\right)}$.
(It does depend on $j^{\left(\mathbf{a}\right)}$, but we will compare
$\mathbf{a}$ to alternative actions $\mathbf{a}'$ with $j^{\left(\mathbf{a}\right)}=j^{\left(\mathbf{a}'\right)}$.)
In particular, for fixed choice of $j^{\left(\mathbf{a}\right)}$,
Alice's expected $\mathbf{x}$-utility against $\overline{{\cal B}}$
is equivalent to 
\begin{equation}
\E_{\mathbf{b}\sim\overline{{\cal B}}}\left[U_{\mathbf{x}}^{A}\left(\mathbf{a};\mathbf{b}\right)\right]=-V-\E_{i\in\sigma_{j^{\left(\mathbf{a}\right)}}}\left[\left(\hat{\mathbf{x}}^{\overline{{\cal B}}}\mid_{i}-\mathbf{x}^{\left(\mathbf{a}\right)}\mid_{i}\right)^{2}\right]=-V-\left\Vert \mathbf{x}^{\left(\mathbf{a}\right)}-\mathbf{\hat{x}}^{\overline{{\cal B}}}\mid_{\sigma_{j^{\left(\mathbf{a}\right)}}}\right\Vert _{2}^{2}.\label{eq:without-variance}
\end{equation}

Let $\mathbf{a}^{*}$ be such that $\mathbf{x}^{\left(\mathbf{a}^{*}\right)}$
that minimizes the distance to $\mathbf{\hat{x}}^{\overline{{\cal B}}}\mid_{\sigma_{j^{\left(\mathbf{a}\right)}}}$
(and otherwise $\mathbf{a}^{*}=\mathbf{a}$ on all other choices).
In particular, we have
\[
\E_{\mathbf{b}\sim\overline{{\cal B}}}\left[U_{\mathbf{x}}^{A}\left(\mathbf{a}^{*};\mathbf{b}\right)\right]\geq-V-\epsilon_{\textsc{Precision}}^{2}.
\]
Alice is really facing distribution ${\cal B}$, which is at distance
$\epsilon_{\textsc{Uniform}}$ from ${\cal B}$. With respect to ${\cal B}$,
we have 
\begin{equation}
\E_{\mathbf{b}\sim{\cal B}}\left[U_{\mathbf{x}}^{A}\left(\mathbf{a}^{*};\mathbf{b}\right)\right]\geq-V-\epsilon_{\textsc{Precision}}^{2}-O\left(\epsilon_{\textsc{Uniform}}\right).\label{eq:a*}
\end{equation}
Since $\mathbf{a}$ and $\mathbf{a}^{*}$ only differ on their $\mathbf{x}$-utility,
if $\mathbf{a}\in\supp\left({\cal A}\right)$ it must be that
\begin{align*}
\E_{\mathbf{b}\sim\overline{{\cal B}}}\left[U_{\mathbf{x}}^{A}\left(\mathbf{a};\mathbf{b}\right)\right] & \geq\E_{\mathbf{b}\sim{\cal B}}\left[U_{\mathbf{x}}^{A}\left(\mathbf{a};\mathbf{b}\right)\right]-O\left(\epsilon_{\textsc{Uniform}}\right)\\
 & \geq\E_{\mathbf{b}\sim{\cal B}}\left[U_{\mathbf{x}}^{A}\left(\mathbf{a}^{*};\mathbf{b}\right)\right]-\epsilon_{\textsc{Nash}}/\lambda_{\mathbf{x}}-O\left(\epsilon_{\textsc{Uniform}}\right)\\
 & \geq-V-\epsilon_{\textsc{Precision}}^{2}-\epsilon_{\textsc{Nash}}/\lambda_{\mathbf{x}}-O\left(\epsilon_{\textsc{Uniform}}\right).
\end{align*}
(Here, the first inequality follows by \autoref{claim:nash-uniform_i,j},
the second inequality by definition of $\epsilon_{\textsc{Nash}}$-WSNE,
and the third from \eqref{eq:a*}.)

Finally, plugging back into \eqref{eq:without-variance}, we have
that
\[
\left\Vert \mathbf{x}^{\left(\mathbf{a}\right)}-\mathbf{\hat{x}}^{\overline{{\cal B}}}\mid_{\sigma_{j^{\left(\mathbf{a}\right)}}}\right\Vert _{2}^{2}\leq\epsilon_{\textsc{Precision}}^{2}+\epsilon_{\textsc{Nash}}/\lambda_{\mathbf{x}}+O\left(\epsilon_{\textsc{Uniform}}\right).
\]

The statement for $\mathbf{x}^{\left(\mathbf{b}\right)}$ follows
analogously.
\end{proof}
\begin{claim}
\label{claim:nash-A-B}$\left\Vert \mathbf{x}^{{\cal A}}-\mathbf{x}^{{\cal B}}\right\Vert _{2}^{2}=O\left(\epsilon_{\textsc{Brouwer}}\right)$
and $\left\Vert \mathbf{x}^{{\cal A}}-\hat{\mathbf{x}}^{{\cal B}}\right\Vert _{2}^{2}=O\left(\epsilon_{\textsc{Brouwer}}\right)$.
\end{claim}

\begin{proof}
Follows by triangle inequality from \autoref{claim:=00005Cbar=00007BA=00007D-A}
and \autoref{claim:pure}.
\end{proof}
\begin{claim}
\label{claim:nash-vertex}We have the following guarantees, depending
on the distance of $\hat{\mathbf{x}}_{r}^{{\cal \overline{\mathcal{B}}}}$
to the nearest encoding of a vertex.
\begin{itemize}
\item If $\hat{\mathbf{x}}_{r}^{{\cal \overline{\mathcal{B}}}}$ is $9\sqrt{h}$-close
to the encoding $\Enc_{C}\left(v\right)$ of some vertex $v\in\left\{ 0,1\right\} ^{n}$,
then $v_{r}^{\left(\mathbf{a}\right)}=v\mid_{\left[n/2\right]}$ for
every $\mathbf{a}\in\supp\left({\cal A}\right)$ and $v_{r}^{\left(\mathbf{b}\right)}=v\mid_{\left[n\right]\setminus\left[n/2\right]}$
for every $\mathbf{b}\in\supp\left({\cal B}\right)$.
\item If $\left\Vert \hat{\mathbf{x}}_{r}^{{\cal \overline{\mathcal{B}}}}-\Enc_{C}\left(v\right)\right\Vert _{2}\in\left[9\sqrt{h},24\sqrt{h}\right]$
for some $v\in\left\{ 0,1\right\} ^{n}$, then $v_{r}^{\left(\mathbf{a}\right)}\in\left\{ v\mid_{\left[n/2\right]},\perp\right\} $
for every $\mathbf{a}\in\supp\left({\cal A}\right)$ and $v_{r}^{\left(\mathbf{b}\right)}\in\left\{ v\mid_{\left[n\right]\setminus\left[n/2\right]},\perp\right\} $
for every $\mathbf{b}\in\supp\left({\cal B}\right)$.
\item If $\hat{\mathbf{x}}_{r}^{{\cal \overline{\mathcal{B}}}}$ is $24\sqrt{h}$-far
from the encoding $\Enc_{C}\left(v\right)$ of every vertex $v\in\left\{ 0,1\right\} ^{n}$,
then either $v_{r}^{\left(\mathbf{a}\right)}=\perp$ for every $\mathbf{a}\in\supp\left({\cal A}\right)$
or $v_{r}^{\left(\mathbf{b}\right)}=\perp$ for every $\mathbf{b}\in\supp\left({\cal B}\right)$. 
\end{itemize}
\end{claim}

\begin{proof}
We prove the first bullet for $\mathbf{\mathbf{a}\in\supp\left({\cal A}\right)}$.
If $\hat{\mathbf{x}}_{r}^{{\cal \overline{\mathcal{B}}}}$ is $9\sqrt{h}$-close
to the encoding $\Enc_{C}\left(v\right)$ of some vertex $v=\left(v^{A}\circ v^{B}\right)\in\left\{ 0,1\right\} ^{n}$,
then by \autoref{claim:nash-A-B} $\mathbf{x}_{r}^{{\cal A}}$ is
$10\sqrt{h}$-close to $\Enc_{C}\left(v\right)$. Hence in particular
$\mathbf{x}_{r}^{{\cal A}}\mid_{\left[m/2\right]}$ is $20\sqrt{h}$-close
to $\Enc_{C}\left(v\right)\mid_{\left[m/2\right]}=\Enc_{C'}\left(v^{A}\right)$.
Thus by \autoref{claim:=00005Cbar=00007BA=00007D-A}, $\mathbf{x}_{r}^{{\cal \overline{A}}}\mid_{\left[m/2\right]}$
is $20\sqrt{h}+O\left(\sqrt{\epsilon_{\textsc{Uniform}}}\right)$-close
to $\Enc_{C'}\left(v^{A}\right)$. By $k$-wise Chernoff (
\autoref{thm:chernoff-k-wise}), we have that with high probability over
random choice of $j^{\left(\mathbf{b}\right)}$, $\mathbf{x}^{{\cal \overline{A}}}\mid_{r\times\left(\tau_{j^{\left(\mathbf{b}\right)}}\cap\left[m/2\right]\right)}$
is also $20\sqrt{h}+O\left(\sqrt{\epsilon_{\textsc{Uniform}}}\right)+o\left(1\right)$-close
to $\Enc_{C'}\left(v^{A}\right)$. If we instead pick $j^{\left(\mathbf{b}\right)}$
according to ${\cal B}$, this holds with probability $1-O\left(\epsilon_{\textsc{Uniform}}\right)$
(using \autoref{claim:nash-uniform_i,j}). And whenever this is
the case, we have by \autoref{claim:pure} that 
\[
\left\Vert \mathbf{x}_{r}^{\left(\mathbf{b}\right)}-\Enc_{C'}\left(v^{A}\right)\mid_{\tau_{j^{\left(\mathbf{b}\right)}}}\right\Vert _{2}\leq20\sqrt{h}+O\left(\sqrt{\epsilon_{\textsc{Brouwer}}}\right)<21\sqrt{h}.
\]
Therefore, Alice's expected $v_{r}$-utility for setting $v_{r}^{\left(\mathbf{a}\right)}=v$
is 
\[
\E_{\mathbf{b}\sim{\cal B}}\left[U_{v_{r}}^{A}\left(v^{\left(\mathbf{a}\right)};j^{\left(\mathbf{b}\right)},\mathbf{x}^{\left(\mathbf{b}\right)}\right)\right]\geq1-O\left(\epsilon_{\textsc{Uniform}}\right),
\]
while for any other choice it is at most $\max\left\{ -1+O\left(\epsilon_{\textsc{Uniform}}\right),0\right\} \ll1$.
Other than her $v_{r}$-utility, the choice of $v_{r}^{\left(\mathbf{a}\right)}$
only affects Alice's $\alpha_{r}$-utility. The claim follows by $\lambda_{v_{r}}/2>\epsilon_{\textsc{Nash}}+\lambda_{\alpha_{r}}$.
The claim for $\mathbf{b}\in\supp\left({\cal B}\right)$ follows analogously
(using $\lambda_{v_{r}}/2>\epsilon_{\textsc{Nash}}+9\lambda_{\hat{\mathbf{x}}}$). 

For the third bullet, if $\hat{\mathbf{x}}_{r}^{{\cal \overline{\mathcal{B}}}}$
is $24\sqrt{h}$-far from $\Enc_{C}\left(v\right)$, then either $\mathbf{x}_{r}^{{\cal A}}\mid_{\left[m/2\right]}$
is $23\sqrt{h}$-far from every $\Enc_{C'}\left(v^{A}\right)$, or
$\mathbf{x}_{r}^{{\cal A}}\mid_{\left[m\right]\setminus\left[m/2\right]}$
is $23\sqrt{h}$-far from every $\Enc_{C'}\left(v^{B}\right)$. Assume
the former wlog. Then by arguments analogous to the first bullet,
Alice's expected $v_{r}$-utility is much higher whenever she sets
$v_{r}^{\left(\mathbf{a}\right)}=\perp$.

Finally, for the second bullet, either $\perp$ or $\Enc_{C'}\left(v^{A}\right)$
may yield Alice a higher expected utility, or they may be about the
same. But $\perp$ has utility at least $0$, which is much higher
than $\approx-1$ that she would get for guessing a wrong $\hat{v^{A}}$.
\end{proof}

\subsubsection*{Implementing the lifting gadget}

The following is an immediate corollary of \autoref{claim:nash-vertex}
(note in particular that $\alpha^{\left(\mathbf{a}\right)}$ only
affects Alice's $\alpha$-utility).
\begin{claim}
\label{claim:nash-alpha}We have the following guarantees, for every
$\mathbf{a}\in\supp\left({\cal A}\right)$, depending on the distance
of $\hat{\mathbf{x}}_{r}^{{\cal \overline{\mathcal{B}}}}$ to the
nearest encoding of a vertex.
\end{claim}

\begin{itemize}
\item If $\hat{\mathbf{x}}_{r}^{{\cal \overline{\mathcal{B}}}}$ is $9\sqrt{h}$-close
to the encoding $\Enc_{C}\left[v\right]$ of some vertex $v\in\left\{ 0,1\right\} ^{n}$,
then $\alpha^{\left(\mathbf{a}\right)}=\alpha\left(v\right)$.
\item If $\left\Vert \hat{\mathbf{x}}_{r}^{{\cal \overline{\mathcal{B}}}}-\Enc_{C}\left(v\right)\right\Vert _{2}\in\left[9\sqrt{h},24\sqrt{h}\right]$
for some $v\in\left\{ 0,1\right\} ^{n}$, then $\alpha^{\left(\mathbf{a}\right)}\in\left\{ \alpha\left(v\right),\perp\right\} $.
\item if $\hat{\mathbf{x}}_{r}^{{\cal \overline{\mathcal{B}}}}$ is $24\sqrt{h}$-far
from the encoding $\Enc_{C}\left[v\right]$ of every vertex $v\in\left\{ 0,1\right\} ^{n}$,
then $\alpha^{\left(\mathbf{a}\right)}=\perp$.
\end{itemize}

\subsubsection*{Locally computing the Brouwer function}
\begin{claim}
\label{claim:nash-f}In expectation over a random pair $\left(\mathbf{a},\mathbf{b}\right)\sim\overline{{\cal A}}\times\overline{{\cal B}}$,
we have that 
\begin{equation}
\E_{\mathbf{a},\mathbf{b}}\left[\left\Vert f\left(\hat{\mathbf{x}}^{\overline{{\cal B}}}\right)\mid_{\sigma_{j^{\left(\mathbf{a}\right)}}\cap\tau_{j^{\left(\mathbf{b}\right)}}}-f^{\left(\mathbf{a},\mathbf{b}\right)}\left(\mathbf{x}^{\left(\mathbf{a}\right)}\right)\right\Vert _{2}^{2}\right]=O\left(\epsilon_{\textsc{Brouwer}}\right).\label{eq:A-f(A)}
\end{equation}
\end{claim}

\begin{proof}
In fact, we prove a stronger claim; namely that for every $\mathbf{a}\in\supp\left({\cal A}\right)$,
we have that with probability at least $1-O\left(\epsilon_{\textsc{Uniform}}\right)$
over $\mathbf{b}\sim\overline{{\cal B}}$, 
\[
\left\Vert f\left(\hat{\mathbf{x}}^{\overline{{\cal B}}}\right)\mid_{\sigma_{j^{\left(\mathbf{a}\right)}}\cap\tau_{j^{\left(\mathbf{b}\right)}}}-f^{\left(\mathbf{a},\mathbf{b}\right)}\left(\mathbf{x}^{\left(\mathbf{a}\right)}\right)\right\Vert _{2}^{2}=O\left(\epsilon_{\textsc{Brouwer}}\right).
\]

By \autoref{claim:nash-vertex} and \autoref{claim:nash-alpha}, it
follows that for $r\in\left\{ 1,2\right\} $, the discrete inputs
$v_{r},P\left(v_{r}\right),S\left(v_{r}\right)$ to the algorithm
from \autoref{lem:doubly-local} are always computed correctly (for
every $\mathbf{b}\in\supp\left({\cal B}\right)$). 

Furthermore, $\mathbf{x}^{\left(\mathbf{a}\right)}$ is used as the
third input to the algorithm; by \autoref{claim:nash-A-B}, we have
that for every $\mathbf{a}\in\supp\left({\cal A}\right)$,
\begin{equation}
\left\Vert \hat{\mathbf{x}}^{\overline{{\cal B}}}\mid_{\sigma_{j^{\left(\mathbf{a}\right)}}}-\mathbf{x}^{\left(\mathbf{a}\right)}\right\Vert _{2}^{2}=O\left(\epsilon_{\textsc{Brouwer}}\right).\label{eq:nash-A-a}
\end{equation}
Thus by $k$-wise Chernoff bound, we also have that with probability
$1-o\left(1\right)$ over a random $\mathbf{b}\sim\overline{{\cal B}}$,
\[
\left\Vert \hat{\mathbf{x}}^{\overline{{\cal B}}}\mid_{\sigma_{j^{\left(\mathbf{a}\right)}}\cap\tau_{j^{\left(\mathbf{b}\right)}}}-\mathbf{x}^{\left(\mathbf{a}\right)}\mid_{\sigma_{j^{\left(\mathbf{a}\right)}}\cap\tau_{j^{\left(\mathbf{b}\right)}}}\right\Vert _{2}^{2}=\left\Vert \hat{\mathbf{x}}^{\overline{{\cal B}}}\mid_{\sigma_{j^{\left(\mathbf{a}\right)}}}-\mathbf{x}^{\left(\mathbf{a}\right)}\right\Vert _{2}^{2}+o\left(1\right)=O\left(\epsilon_{\textsc{Brouwer}}\right).
\]
Therefore with probability $1-O\left(\epsilon_{\textsc{Uniform}}\right)$,
the same also holds for a random $\mathbf{b}\sim{\cal B}$. Whenever
this is the case, we have by \autoref{lem:doubly-local} that indeed
\[
\left\Vert f\left(\hat{\mathbf{x}}^{\overline{{\cal B}}}\right)\mid_{\sigma_{j^{\left(\mathbf{a}\right)}}\cap\tau_{j^{\left(\mathbf{b}\right)}}}-f^{\left(\mathbf{a},\mathbf{b}\right)}\left(\mathbf{x}^{\left(\mathbf{a}\right)}\right)\right\Vert _{2}^{2}=O\left(\epsilon_{\textsc{Brouwer}}\right).
\]

Therefore, by \autoref{claim:nash-uniform_i,j} the same holds for
a random $\mathbf{b}\sim\overline{{\cal B}}$, which completes the
proof.
\end{proof}

\subsubsection*{Completing the proof}
\begin{proof}
[Proof of \autoref{prop:nash-A-f(A)}] Let $\overline{{\cal A}}\left(i\right)$
denote the restriction of $\overline{{\cal A}}$ to actions $\mathbf{a}$
such that $\sigma_{j^{\left(\mathbf{a}\right)}}\ni i$; define $\overline{{\cal B}}\left(i\right)$
analogously (conditioning on $\tau_{j^{\left(\mathbf{b}\right)}}\ni i$).
Similarly to Equation \eqref{eq:x^A-decompose}, Bob's $\hat{\mathbf{x}}$-utility
facing an expected $\mathbf{a}\sim\overline{{\cal A}}$ decomposes
as
\[
\E_{\mathbf{a\sim\overline{{\cal A}}}}\left[U_{\hat{\mathbf{x}}}^{B}\left(\mathbf{b};\mathbf{a}\right)\right]=-\E_{i\in\tau_{j^{\left(\mathbf{b}\right)}}}\left[\underbrace{\Var_{\mathbf{a}\sim\overline{{\cal A}}\left(i\right)}\left(f_{i}^{\left(\mathbf{a},\mathbf{b}\right)}\left(\mathbf{x}^{\left(\mathbf{a}\right)}\right)\right)}_{\eqqcolon V}+\left(\hat{\mathbf{x}}^{\mathbf{b}}\mid_{i}-\E_{\mathbf{a}\sim\overline{{\cal A}}\left(i\right)}\left[f_{i}^{\left(\mathbf{a},\mathbf{b}\right)}\left(\mathbf{x}^{\left(\mathbf{a}\right)}\right)\right]\right)^{2}\right].
\]
Now the variance term $V$ depends on other parts of Bob's action
(in particular $j^{\left(\mathbf{b}\right)},v^{\left(\mathbf{b}\right)}$),
but not on the choice of $\hat{\mathbf{x}}^{\left(\mathbf{\mathbf{b}}\right)}$.
Let $\mathbf{x}^{*}\in\left[-1,2\right]^{T_{j^{\left(\mathbf{b}\right)}}}$
denote the optimal choice of $\hat{\mathbf{x}}^{\left(\mathbf{\mathbf{b}}\right)}$
(fixing the rest of Bob's action) when facing mixed strategy $\overline{{\cal A}}$. 

For every $i\in\tau_{j^{\left(\mathbf{b}\right)}}$, the optimal $x_{i}^{*}$
is within $\pm\epsilon_{\textsc{Precision}}$ of $\E_{\mathbf{a}\sim\overline{{\cal A}}\left(i\right)}\left[f_{i}^{\left(\mathbf{a},\mathbf{b}\right)}\left(\mathbf{x}^{\left(\mathbf{a}\right)}\right)\right]$.
When he sets $\hat{\mathbf{x}}^{\left(\mathbf{\mathbf{b}}\right)}=\mathbf{\hat{\mathbf{x}}^{*}}$
and faces mixed strategy $\overline{{\cal A}}$, Bob's expected $\hat{\mathbf{x}}$-utility
is at least $-V-\epsilon_{\textsc{Precision}}^{2}$; thus by \autoref{claim:nash-uniform_i,j} when facing strategy ${\cal A}$ his
expected $\hat{\mathbf{x}}$-utility is at least $-V-\epsilon_{\textsc{Precision}}^{2}-O\left(\epsilon_{\textsc{Uniform}}\right)$.
Therefore, for every action in Bob's support (restricted to same choice
of $j^{\left(\mathbf{b}\right)},v^{\left(\mathbf{b}\right)}$), his
expected $\hat{\mathbf{x}}$-utility (when facing either $\overline{{\cal A}}$
or ${\cal A}$) must be at least $-V-\epsilon_{\textsc{Precision}}^{2}-\epsilon_{\textsc{Nash}}/\lambda_{\hat{\mathbf{x}}}-O\left(\epsilon_{\textsc{Nash}}+\lambda_{\mathbf{x}}\right)=-V-O\left(\epsilon_{\textsc{Brouwer}}\right)$.
For any such action, we have that,
\begin{equation}
\E_{i\in\tau_{j^{\left(\mathbf{b}\right)}}}\left[\left(\hat{\mathbf{x}}^{\mathbf{b}}\mid_{i}-\E_{\mathbf{a}\sim\overline{{\cal A}}\left(i\right)}\left[f_{i}^{\left(\mathbf{a},\mathbf{b}\right)}\left(\mathbf{x}^{\left(\mathbf{a}\right)}\right)\right]\right)^{2}\right]=O\left(\epsilon_{\textsc{Brouwer}}\right).\label{eq:b-f(a)}
\end{equation}
We now want to argue that on average $\mathbf{b}\sim{\cal \overline{B}}$
and $i\in\tau_{j^{\left(\mathbf{b}\right)}}$, we also have that $\E_{\mathbf{a}\sim\overline{{\cal A}}\left(i\right)}\left[f_{i}^{\left(\mathbf{a},\mathbf{b}\right)}\left(\mathbf{x}^{\left(\mathbf{a}\right)}\right)\right]$
is close to $f_{i}\left(\hat{\mathbf{x}}^{\overline{{\cal B}}}\right)$.
Indeed,

\begin{align}
\E_{\mathbf{b}\sim\overline{{\cal B}}}\E_{i\in\tau_{j^{\left(\mathbf{b}\right)}}}\left[\left(f_{i}\left(\hat{\mathbf{x}}^{\overline{{\cal B}}}\right)-\E_{\mathbf{a}\sim\overline{{\cal A}}\left(i\right)}\left[f_{i}^{\left(\mathbf{a},\mathbf{b}\right)}\left(\mathbf{x}^{\left(\mathbf{a}\right)}\right)\right]\right)^{2}\right] & \leq\E_{\mathbf{b}\sim\overline{{\cal B}}}\E_{i\in\tau_{j^{\left(\mathbf{b}\right)}}}\E_{\mathbf{a}\sim\overline{{\cal A}}\left(i\right)}\left[\left(f_{i}\left(\hat{\mathbf{x}}^{\overline{{\cal B}}}\right)-f_{i}^{\left(\mathbf{a},\mathbf{b}\right)}\left(\mathbf{x}^{\left(\mathbf{a}\right)}\right)\right)^{2}\right]\label{eq:f(a)-f(A)-1}\\
 & =\E_{\mathbf{a},\mathbf{b}\sim\overline{{\cal A}}\times\overline{{\cal B}}}\left[\left\Vert f\left(\hat{\mathbf{x}}^{\overline{{\cal B}}}\right)\mid_{\sigma_{j^{\left(\mathbf{a}\right)}}\cap\tau_{j^{\left(\mathbf{b}\right)}}}-f^{\left(\mathbf{a},\mathbf{b}\right)}\left(\mathbf{x}^{\left(\mathbf{a}\right)}\right)\right\Vert _{2}^{2}\right]\label{eq:f(a)-f(A)-2}\\
 & =O\left(\epsilon_{\textsc{Brouwer}}\right).\label{eq:f(a)-f(A)-3}
\end{align}
Above, we can change the order of expectations (from \eqref{eq:f(a)-f(A)-1}
to \eqref{eq:f(a)-f(A)-2}) because $\overline{{\cal A}},\overline{{\cal B}}$
have uniform marginals over $j$'s and the intersection of any $\sigma_{j^{\left(\mathbf{a}\right)}},\tau_{j^{\left(\mathbf{b}\right)}}$
has cardinality exactly $\left|\sigma_{j^{\left(\mathbf{a}\right)}}\cap\tau_{j^{\left(\mathbf{b}\right)}}\right|=m/\ell^{2}$.
Then \eqref{eq:f(a)-f(A)-3} follows by \autoref{claim:nash-f}.

Combining \eqref{eq:b-f(a)} with \eqref{eq:f(a)-f(A)-3} via the
triangle inequality, we have that
\[
\E_{\mathbf{b}\sim\overline{{\cal B}}}\E_{i\in\tau_{j^{\left(\mathbf{b}\right)}}}\left[\left(f_{i}\left(\hat{\mathbf{x}}^{\overline{{\cal B}}}\right)-\hat{\mathbf{x}}^{\mathbf{b}}\mid_{i}\right)^{2}\right]=O\left(\epsilon_{\textsc{Brouwer}}\right).
\]
Since $\overline{{\cal B}}$ has uniform $j$-marginals, we can again
replace the order of expectations and write
\begin{align*}
\left\Vert \hat{\mathbf{x}}^{{\cal \overline{B}}}-f\left(\hat{\mathbf{x}}^{\overline{{\cal B}}}\right)\right\Vert _{2}^{2} & =\E_{i}\E_{\mathbf{b}\sim\overline{{\cal B}}\left(i\right)}\left[\left(f_{i}\left(\hat{\mathbf{x}}^{\overline{{\cal B}}}\right)-\hat{\mathbf{x}}^{\mathbf{b}}\mid_{i}\right)^{2}\right]\\
 & =\E_{\mathbf{b}\sim\overline{{\cal B}}}\E_{i\in\tau_{j^{\left(\mathbf{b}\right)}}}\left[\left(f_{i}\left(\hat{\mathbf{x}}^{\overline{{\cal B}}}\right)-\hat{\mathbf{x}}^{\mathbf{b}}\mid_{i}\right)^{2}\right]=O\left(\epsilon_{\textsc{Brouwer}}\right).
\end{align*}

\autoref{prop:nash-A-f(A)} now follows from \autoref{claim:=00005Cbar=00007BA=00007D-A},
$O\left(1\right)$-Lipschitzness of $f$ and the triangle inequality.
\end{proof}

\pagebreak

\addcontentsline{toc}{section}{References}
\DeclareUrlCommand{\Doi}{\urlstyle{sf}}
\renewcommand{\path}[1]{\small\Doi{#1}}
\renewcommand{\url}[1]{\href{#1}{\small\Doi{#1}}}
\bibliographystyle{alphaurl}
\bibliography{cc-nash}

\end{document}